\documentclass[%
 reprint,
 nofootinbib,
superscriptaddress,
 amsmath,amssymb,
 aps,
onecolumn
]{revtex4-2}

\usepackage{graphicx}%
\usepackage{dcolumn}%
\usepackage{bm}%
\usepackage{braket}
\usepackage{amsthm}%
\usepackage{amsmath}
\usepackage{algorithm}
\usepackage{algpseudocode}
\usepackage{tabularx}
\usepackage[inline]{asymptote}
\usepackage{siunitx}
\usepackage{tikz-cd}
\usepackage{todonotes}
\usetikzlibrary{backgrounds}
\usepackage{tcolorbox}
\tcbuselibrary{theorems,skins,hooks,breakable}
\tcbset{myexastyle/.style={
  enhanced,
  breakable,
  colback=gray!10,
  colframe=gray!50!black,
  boxrule=0.5pt,
  fonttitle=\bfseries,
  title={Example \thetcbcounter.\quad #1},
  coltitle=white
}}

\newtcbtheorem[number within=section]{algorithminner}{Example}{myexastyle}{ex}

\newenvironment{exa}[1]
{\begin{algorithminner}{#1}{algorithm-\arabic{section}-[section]}}
{\end{algorithminner}}

\usepackage[colorlinks=true,linkcolor=blue,allcolors=blue]{hyperref}
\usepackage[capitalize,nameinlink]{cleveref}
\usepackage[table]{xcolor}

\newtheorem{theorem}{Theorem}[section]
\newtheorem{lemma}{Lemma}[section]
\newtheorem{definition}{Definition}[section]
\newtheorem{corollary}{Corollary}[section]
\newtheorem{prop}{Proposition}[section]

\theoremstyle{remark} %
\newtheorem{remark}{Remark}[section]

\begin{document}
\title{Copy-cup Gates in Tensor Products of Group Algebra Codes}

\author{Ryan Tiew}
\affiliation{School of Mathematics, Fry Building, Woodland Road, Bristol BS8 1UG, UK}
\affiliation{Quantum Engineering Centre for Doctoral Training, University of Bristol, BS8 1FD, UK}
\author{Nikolas P.\ Breuckmann}
\affiliation{School of Mathematics, Fry Building, Woodland Road, Bristol BS8 1UG, UK}

\date{\today}%

\begin{abstract}
We determine conditions on classical group algebra codes so that they have pre-orientation for cup products and copy-cup gates. This defines quantum codes that have constant-depth $\operatorname{CZ}$ and $\operatorname{CCZ}$ gates constructed via tensor products of classical group algebra codes, including hypergraph and balanced products. We show that determining the conditions relies on solving the perfect matching problem in graph theory. Conditions are fully determined for the 2- and 3-copy-cup gates, for group algebra codes up to weight 4, including for codes with odd check weight. These include the bivariate bicycle codes, which we show do not have the pre-orientation for either type of copy-cup gate. We show that abelian weight 4 group algebra codes satisfying the non-associative 3-copy-cup gate necessarily have a code distance of 2, whereas codes that satisfy conditions for the symmetric 3-copy-cup gate can have higher distances, and in fact also satisfy conditions for the 2-copy-cup gate. Finally we find examples of quantum codes from the product of abelian group algebra codes that have inter-code constant-depth $\operatorname{CZ}$ and $\operatorname{CCZ}$ gates.
\end{abstract}

\maketitle
\section{Introduction}
Quantum error-correcting codes have received sustained interest in recent years due to their necessity in constructing universal fault-tolerant quantum computers. 
Schemes are necessary to implement logic gates in a fault-tolerant manner, and there has been a lot of work in this area in tandem with the construction of high-rate quantum LDPC codes \cite{Breuckmann_2021}. 
For example, there have been a number of schemes for logic gates such as ones that generalise lattice surgery~\cite{Horsman_2012, Cohen_2022,cowtan2026parallellogicalmeasurementsquantum, cross2024improvedqldpcsurgerylogical, swaroop2024universaladaptersquantumldpc, he2025extractorsqldpcarchitecturesefficient}, as well as ones involving homomorphic measurements \cite{huang2022homomorphiclogicalmeasurements, PhysRevX.15.021065, ide2024faulttolerantlogicalmeasurementshomological}.
On the other hand, a different fault-tolerant scheme involves using codes that have non-Clifford gates that can be implemented in a constant-depth circuit by construction, such as those in \cite{he2025asymptoticallygoodquantumcodes,nguyen2024goodbinaryquantumcodes,golowich2024quantumldpccodestransversal}. %
Constant-depth circuits guarantee fault tolerance as qubit errors only propagate sparsely through physical qubits before error correction. 
Codes with high rate and constant-depth non-Clifford gates can then be used, for example, in magic state distillation \cite{menon2025magictricyclesefficientmagic} as part of a larger quantum computer or architecture.

It is a difficult endeavour to design quantum LDPC codes equipped with constant-depth gates. 
Recently, the so-called \emph{$\Lambda$-copy-cup gate} \cite{breuckmann2024cupsgatesicohomology} was defined, which allows for diagonal constant-depth gates on quantum codes that satisfy certain properties, and that can reach any desired level of the Clifford hierarchy. 
The work relies on the extension of the topological cup product towards more general cochain complexes that define quantum codes. 
Since their work, other papers, in particular \cite{menon2025magictricyclesefficientmagic,jacob2025singleshotdecodingfaulttolerantgates} have come up using the 3-copy-cup gate and their conditions to find codes with constant-depth non-Clifford gates. 

In this paper, we build on these existing bodies of work. We determine constraints on classical group algebra codes that result in copy-cup gates on the tensor product quantum codes, and show that the constraints can be determined combinatorially. 
We also extend the work of \cite{breuckmann2024cupsgatesicohomology} towards classical codes with odd check weight, and show that the restriction to codes with even check weight is unnecessary. 
This immediately lets us analyse the bivariate bicycle codes \cite{bravyi2023highthreshold}, which are constructed from abelian group algebra codes of weight 3. 
We strengthen this analysis with a numerical search for codes similar to the bivariate bicycle codes -- that is, quantum codes with $\operatorname{X}$ and $\operatorname{Z}$ check weight 6 -- from abelian groups that are equipped with the 2-copy-cup gate.
We also show they cannot have 3-copy-cup gates. We further show that higher weight classical codes lead to more code constraints. 
We perform a numerical search for abelian group algebra codes with pre-orientation that satisfy the 2- and 3- copy-cup gates up to weight 4, and illustrate how constraints grow with check weight by providing examples of restrictive conditions for weight 6 group algebra codes with pre-orientation for the 3-copy-cup gate.

The rest of the paper is structured as follows. In \cref{sec: background}, we provide relevant background information to understand our work. We define classical group algebra codes and quantum CSS codes, and their relation to cochain complexes. We show how cohomology invariants are related to quantum logic gates, and go through necessary definitions of the copy-cup gate from \cite{breuckmann2024cupsgatesicohomology} in a pedagogical manner, to illustrate the differences in pre-orientations for the copy-cup gates that arise from the non-associativity of the cup product. In \cref{sec: application}, we determine pre-orientation conditions for 3-term group algebra codes, for both the 2- and 3-copy-cup gates, and present results from the numerical search for weight 6 quantum codes that have constant-depth $\operatorname{CZ}$ from the 2-copy-cup gate. In \cref{sec: higher weight}, we generalise this to a combinatorial argument for pre-orientation for codes with higher check weights, and apply them to weight 4 group algebra codes. This is further supplemented with examples of quantum codes with 2-copy-cup gates from a numerical search. In \cref{sec: further}, we present results of targeted numerical searches, including for codes from univariate cyclic groups and general abelian groups, for quantum codes equipped with the 3-copy-cup gate, and use the combinatorial argument to determine pre-orientation conditions for 6-term group algebra codes. Finally, in \cref{sec: conclusion}, we give a summary of the work, and discuss future avenues for exploration.
\section{Background} \label{sec: background}
In this section, we define notation necessary in the remainder of our work. We first provide background information on cochain complexes and how they are related to CSS codes, followed by an explanation of coinvariants and diagonal logic gates, and then a description on classical group algebra codes. Then, we go through redefinitions from \cite{breuckmann2024cupsgatesicohomology} necessary for the copy-cup gate, highlighting important aspects by providing examples to illustrate subtle differences, particularly in the non-associativity of the cup product. 
\subsection{Cochain complexes and CSS codes} \label{sec: cochains}
Let us first define a cochain complex and how they are related to error-correcting codes.
\begin{definition}[Cochain complex]
    Let $R$ be a ring. A \emph{cochain complex} $A$ is a sequence of $R$-modules $\{A^i\}$ and linear maps $\delta^i$, known as coboundary operators, and written as
    \begin{center}
    \begin{tikzcd}[cells={nodes={minimum height=2em}}]
    A = ( \cdots \arrow[r, "\delta^{j-2}"] & A^{j-1} \arrow[r, "\delta^{j-1}"] & A^j\arrow[r, "\delta^j"] & A^{j+1} \arrow[r, "\delta^{j+1}"] & \cdots ),
    \end{tikzcd}
    \end{center}
    such that the maps satisfy the relation
    \begin{equation}
        \delta^{i+1}\circ \delta^i = 0.
    \end{equation}
    The complex is said to be \emph{based} if it is equipped with a basis $X_i^A$ for each $A^i$; and when the context is clear, we omit the superscript. Elements of $A^i$ are known as $i$-cochains, while $\operatorname{im}\delta^{i-1}$ and $\ker \delta^i$ are known as \emph{$i$-coboundaries} and \emph{$i$-cocycles} respectively. Then, the $i^{th}$ cohomology group is defined to be
    \begin{equation}
        \mathcal{H}^i = \ker \delta^{i}/\operatorname{im} \delta^{i-1}.
    \end{equation}
    Dually, we can also define chain complexes and homology by taking boundary operators $\partial_{i+1}:=(\delta^{i})^T$. Boundary operators satisfy the conditions
    \begin{equation}
        \partial_i\circ\partial_{i+1}=0
    \end{equation}
    and homology can be similarly defined as
    \begin{equation}
        \mathcal{H}_i = \ker \partial_i/\operatorname{im} \partial_{i+1}.
    \end{equation}    
\end{definition}
In this work, we consider based cochain complexes defined over the binary field $\mathbb{F}_2$. We note in particular that a basis $X_0$ of 0-cochains are known as \emph{vertices} while a basis $X_1$ of 1-cochains are \emph{edges}. 
In the rest of the paper, when the meaning is clear, we omit the superscript, and refer to a basis of $i$-cochains as $X_i$.

We can also define a stabiliser quantum error-correcting code as one whose codewords are defined by the set of states
\begin{equation}
    \mathcal{C}=\{\ket{\psi}\mid s\ket{\psi}=\ket{\psi}\forall s\in \mathcal{S}\},
\end{equation}
where $\mathcal{S}$ is an abelian subgroup of the Pauli group $\operatorname{P}_n$ that does not contain $-\operatorname{I}$~\cite{GOTTESMAN2006196}. A CSS code is a type of stabiliser quantum error-correcting code where $\mathcal{S}$ can be generated by tensor products of solely Pauli $\operatorname{X}$ or $\operatorname{Z}$. CSS codes have parameters $[[n,k,d]]$, where $n$ is the number of physical qubits of the code, $k$ is the number of logical qubits, and $d$ is the code distance. They are defined by two parity check matrices $H_X$ and $H_Z$ whose rows encode the $\operatorname{X}(\operatorname{Z})$ stabilisers of the code. The check matrices must satisfy the commutation relation
\begin{equation}
    H_XH_Z^T = 0.
\end{equation}
By identifying $H_X$ with $(\delta^{i})^T$, and $H_Z$ with $\delta^{i+1}$,
we see that a CSS code is a 3-term cochain complex built from $\{A^{i+1},A^i,A^{i-1}\}$ for some index $i$ and their respective coboundary operators that can be excised from or straightforwardly identified with $A$, with physical qubits identified with the basis $X_i^A$. Logical operators $\bar{\operatorname{X}}$ and $\bar{\operatorname{Z}}$ are identified with the cohomology and homology classes of the complex respectively. In the rest of this work, we take physical qubits to be defined on the level of 1-cochains, i.e. $i=1$.
\begin{remark}
While CSS codes only need to be defined on 3-term cochain complexes, in some cases there are benefits to use longer complexes. For example, 5-term cochain complexes can be constructed with the additional terms $A^{i+2},A^{i-2}$ used as \emph{metachecks} on the syndrome of the quantum code \cite{Campbell_2019}. In \cite{jacob2025singleshotdecodingfaulttolerantgates}, 4-term cochain complexes give rise to \emph{single-shot decodability} for errors in the $\operatorname{X}$ basis, with the additional cochain term used to define $\operatorname{Z}$-metachecks. The work of \cite{breuckmann2024cupsgatesicohomology} also utilise longer complexes, which are necessary in the definition of the cup product, to implement constant-depth copy-cup gates ascending the Clifford hierarchy.
\end{remark}
\subsection{Cohomology invariants and diagonal quantum gates}\label{sec: cohom invariants}
In defining the quantum code on a cochain complex, there is a natural relation between \emph{cohomology invariants} and quantum logic gates. As we have defined the cochain complexes over $\mathbb{F}_2$, a cohomology invariant is a map from cohomology classes to an element in $\mathbb{F}_2$. Operations on $i$-cochains
\begin{equation}
    \psi: A^i\rightarrow \mathbb{F}_2
\end{equation}
are cohomology invariants if, for $i$-cocycles $\gamma,\gamma'\in \ker \delta^i\subset A^i$, they implement
\begin{equation}
    \psi(\gamma) = \psi(\gamma'):=\psi([\gamma]) \qquad \text{if } \gamma,\gamma'\in[\gamma],
\end{equation}
where $[\gamma]$ is the cohomology class corresponding to $\gamma$. In other words, the operation $\psi$ depends solely on cohomology classes. The operation $\psi$ is said to \emph{descend} to a cohomology invariant, implementing a map $\psi: \mathcal{H}^i\rightarrow \mathbb{F}_2$.

Let us now make the correspondence between cohomology classes and logical basis codestates. A 3-term cochain complex with cohomology group $\mathcal{H}^i$ of dimension $k$ has $k$ cohomology classes, given as $\{[\gamma_1],\cdots,[\gamma_k]\}$, where $\gamma_i$ are representative cocycles $\gamma_i\in \ker\delta^i,i\in\{1,\cdots,k\}$. Then, in the $[[n,k,d]]$ CSS code identified with the complex, each $\gamma_i$ corresponds to a logical operator $\bar{\operatorname{X}}_i$, while its cohomology class $[\gamma_i]$ corresponds to $\bar{\operatorname{X}}_i$ up to addition of $\operatorname{X}$ stabilisers. A logical computational basis state that is $1$ on the $i^{th}$ logical qubit can be written as
\begin{equation}
    \sum_{\substack{
    s_x\in \operatorname{rowspace}(H_X)\\
    s_z\in \operatorname{rowspace}(H_Z)
    }
    }
    \left(\prod_{j\in \operatorname{supp}(s_x)}\operatorname{X}_j\right)\left(\prod_{j'\in \operatorname{supp}(s_z)}\operatorname{Z}_{j'}\right) \bar{\operatorname{X}}_i\ket{0}^{\otimes n}
\end{equation}
up to normalisation factors, where $s_x,s_z$ are $\mathbb{F}_2^n$ vectors, and $\operatorname{supp}(s_x), \operatorname{supp}(s_z)$ indicates their support. Recalling that the basis of $i$-cochains is identified with qubits, we can identify the support of $\bar{\operatorname{X}}_i$ with that of the cocycle $\gamma_i$; and by identifying the $\operatorname{X}$ stabilisers with $\operatorname{im}\delta^i$, we see that this codestate can be identified with the cohomology class $[\gamma_i]$ as
\begin{equation}
    \ket{[\gamma_i]}:=\sum_{\gamma'\in[\gamma_i]} \ket{\gamma'}.
\end{equation}
Cohomology invariants can then be used to construct \emph{diagonal} quantum logic gates, which act on logical computational basis states to implement a phase. They are used to construct a unitary $U_\psi$ which preserves the codespace, and that implements a phase
\begin{equation}
    U_\psi\ket{[\gamma_i]} = e^{i\pi\psi([\gamma_i])}\ket{[\gamma_i]}.
\end{equation}
Such a unitary preserves the codespace as it implements a phase $e^{i\pi\phi([\gamma])}$ on all states in the superposition. Then, when we consider an \emph{arbitrary} logical codestate written in the logical computational basis
\begin{equation}
    \ket{\chi} = \sum_{i=1}^k \alpha_i \ket{[\gamma_i]}, \qquad \alpha_i\in \mathbb{C} \text{ such that }\sum_{i=1}^{k}|\alpha_i|^2=1,
\end{equation}
the unitary action is extended to the codestate by linearity and given as
\begin{align}
    U_{\psi}\ket{\chi} &= \sum_{i=1}^k \alpha_i U_{\psi}\ket{[\gamma_i]}\\
    &= \sum_{i=1}^k\alpha_ie^{i\pi\phi([\gamma_i])}\ket{[\gamma_i]}.
\end{align}
In turn, $U_\psi$ has a representation as a $k\times k$ matrix with diagonal entries $e^{i\pi\phi([\gamma_i])}$, with basis the logical computational basis, and is therefore said to be a diagonal quantum logic gate.
\subsection{Classical codes and group algebra codes} \label{sec: group algebra}
A binary classical code can be identified as the linear vector subspace $\ker H$, where $H$ is the \emph{parity check matrix} of the code defined over $\mathbb{F}_2$. Similar to quantum codes, classical codes have parameters $[n,k,d]$ that denote the code length, number of encoded bits, and distance respectively. Following the previous section, by taking $H:=(\delta^i)^T$, we then see that classical codes can be defined as the homology group of a 2-term cochain complex. When $i=0$, edges are information bits, while vertices are parity checks.

The group algebra $KG$ for group $G$ with cardinality $|G|$ over field $K$ is the set of all linear combinations of finitely-many elements $g\in G$ with coefficients in $K$. In this work, we are interested in group algebra codes \cite{Berman_1969,6770805} over the binary field $\mathbb{F}_2$. Then, the group algebra is a free $\mathbb{F}_2$-vector space with basis $G$. A group algebra element $\alpha\in \mathbb{F}_2[G]$ can be represented by
\begin{equation}
    \alpha=\sum_{g\in G} \alpha_g g,\, \alpha_g\in\mathbb{F}_2.
\end{equation}
Addition of two group algebra elements $\alpha,\beta$ can be written as
\begin{subequations}
    \begin{align}
    \alpha+\beta &= \sum_{g\in G} \alpha_gg+\sum_{g\in G}\beta_gg\\
    &=\sum_{g\in G}(\alpha_g+\beta_g)g
    \end{align}
\end{subequations}
while multiplication can be written as 
\begin{subequations}
    \begin{align}
    \alpha\beta &= \left(\sum_{g\in G} \alpha_gg\right)\left(\sum_{h\in G}\beta_hh\right)\\
    &=\sum_{g,h\in G}\left(\alpha_g\beta_h\right) gh,
    \end{align}
\end{subequations}
with multiplication of group elements defined by the group operation. Group algebra elements can be converted into a parity check matrix $H$ by taking their \emph{left regular representation} over $\mathbb{F}_2$, where every group element is mapped to a square $|G| \times |G|$ matrix representing its permutation action on other group elements via multiplication from the left. We can therefore define classical group algebra codes as the kernel of a group algebra element mapped to its regular representation. Following the notation of \cite{jitman2010checkablecodesgrouprings}, we define this group algebra element to be the \textit{check element} of the group algebra code, and freely identify the check element with its matrix representation.
Since individual group elements are mapped to permutation matrices, the parity check matrix defined by the check element $\alpha\in\mathbb{F}_2[G]$ has $|\alpha|$ non-zero entries in every row and column, with $|\cdot|$ defined as the Hamming weight of the group algebra element as a binary vector; we denote this as the \emph{check weight} of the code, and define \emph{$m$-term} or \emph{weight $m$ group algebra codes} as the kernel of a check element with $m$ distinct group elements. Note also that multiplication of $\alpha$ by $g\in G$ has the property $|\alpha g|=|\alpha|$ by the group rearrangement theorem.

Following the discussion in \cref{sec: cochains}, we can also define group algebra codes on 2-term cochain complexes. Group algebra codes with check element $\alpha$ can be defined by the homology group of the 2-term cochain complex
    \begin{center}
    \begin{tikzcd}[cells={nodes={minimum height=2em}}]
        R \arrow[r, "\delta\in R"] & R, R:=\mathbb{F}_2[G], \delta:=\alpha^{-1},
    \end{tikzcd}
    \end{center}
where the superscript $^{-1}$ indicates the \emph{antipode map} which maps $g$ to its inverse $g^{-1}$ for every group element $g$ in the check element $\alpha$. Equivalently, when considering the matrix representations of the check element and coboundary operator, the map corresponds to the matrix transpose on the regular representation of the check element. The group $G$ also constitutes bases $X_0,X_1$ for 0- and 1-cochains of the complex respectively. To avoid cumbersome notation in referring to group algebra codes by $\alpha$ and coboundary operators by $\alpha^{-1}$, for the rest of the paper we refer to a group algebra code by its coboundary operator $\delta$.

When $G$ is abelian, group algebras correspond to polynomial rings, and the check element is referred to as a \emph{check polynomial}. For example, the univariate cyclic group algebra presented as
\begin{equation}\notag
\mathbb{F}_2[G] = \mathbb{F}_2\langle x\mid x^n=1\rangle
\end{equation}
is in one-to-one correspondence with the polynomial quotient ring
\begin{equation}\notag
\mathbb{F}_2[x]/\langle x^n+1\rangle.
\end{equation}
Examples of abelian group algebra codes include repetition codes, Hamming codes, and BCH codes.
\begin{remark}
    Quantum codes constructed as the tensor product of abelian group algebra codes include the toric code and its higher-dimensional variants, and bivariate bicycle codes \cite{bravyi2023highthreshold} and their generalisations, e.g.\ \cite{pryadko_2BGA,mian2026multivariatemulticyclecodescomplete, lin2025abelianmulticyclecodessingleshot,liang2025generalizedtoriccodestwisted,symons2025sequencesbivariatebicyclecodes, tiew2024lowoverheadentanglinggatesgeneralised}. 
\end{remark}
\begin{exa}{Repetition Code}
    The $[n,1,n]$ classical repetition code is a group algebra code from a cyclic group algebra $\mathbb{F}_2[C_n]$, where $C_n$ is generated by the element $x$. The group algebra element $1+x$ defines the check polynomial of the repetition code. We can convert the check polynomial to an $n\times n$ parity check matrix by mapping individual terms of the check polynomial to its regular representation; in this case, we have
    \begin{align}
        x &\mapsto \begin{pmatrix} 
        0 & 0 & 0 & \dots & 1 \\
        1 & 0 & 0 & \dots & 0 \\
        0 & 1 & 0 & \dots & 0 \\
        \vdots & \vdots & \ddots & \vdots & \vdots \\
        0 & 0 & \dots & 1 & 0 
        \end{pmatrix},\\
        1 &\mapsto \begin{pmatrix} 
        1 & 0 & 0 & \dots & 0 \\
        0 & 1 & 0 & \dots & 0 \\
        0 & 0 & 1 & \dots & 0 \\
        \vdots & \vdots & \vdots & \ddots & \vdots \\
        0 & 0 & 0 & \dots & 1 
        \end{pmatrix}.
    \end{align}
\end{exa}
\subsection{Tensor product cochain complexes}
The tensor product gives us a method to construct longer cochain complexes from shorter ones; and therefore to construct quantum codes from classical codes. 
The tensor product over the binary field $\mathbb{F}_2$ gives the so-called hypergraph product \cite{Tillich_2014}, while tensor products over more general rings defines the balanced product \cite{Breuckmann_2021_2}. 
We focus on taking balanced products over group algebras $R=\mathbb{F}_2[G]$. The balanced product of two vector spaces $V,W$ equipped with a linear right (resp. left) action of $G$ over $R$ is the quotient
\begin{equation}
    V\otimes_{R}W = V\otimes W/\langle vg\otimes w - v\otimes gw\rangle
\end{equation}
for $v\in V, w\in W, g\in G$. The balanced product partitions elements of the vector space into equivalence classes, which are $G$-orbits, with an anti-diagonal action of $G$. These equivalence classes, explicitly given as 
\begin{equation}
    v\otimes_Rw := \{vg\otimes g^{-1}w\mid g\in G\},
\end{equation}
are also known as \emph{coinvariants}. When the vector spaces $V$ and $W$ are both the group algebra $\mathbb{F}_2[G]$, there is an isomorphism
\begin{equation} \label{eqn: RR isomorphism}
    R\otimes_{R}R \cong R,\, R=\mathbb{F}_2[G]
\end{equation}
by identifying $g_1\otimes_R g_2\mapsto g_1g_2$ for $g_1,g_2\in G$.

In this work, we focus on taking $\Lambda$-fold balanced products of 2-term cochain complexes $\{A_i\}$ with coboundary operators $\{\delta_{A_i}\},i\in \{1,\cdots,\Lambda\}$ corresponding to classical codes. This returns the complex $E$. A \emph{total complex} 
\begin{center}
    \begin{tikzcd}[cells={nodes={minimum height=2em}}]
    \operatorname{Tot}(E) = ( \cdots \arrow[r, "\delta^{n-2}_E"] & E^{n-1} \arrow[r, "\delta^{n-1}_E"] & E^n\arrow[r, "\delta^n_E"] & E^{n+1} \arrow[r, "\delta^{n+1}_E"] & \cdots ),
    \end{tikzcd}
\end{center}
can then be defined by collapsing the product complex along terms with the same grading; that is, we have terms
\begin{equation}
    E^n = \bigoplus_{i_1+\cdots +i_\Lambda =n}A_1^{i_1}\otimes \cdots \otimes A_{\Lambda}^{i_\Lambda}
\end{equation}
and with coboundary operators
\begin{equation}
    \delta^n_E = \bigoplus_{\substack{i_1+\cdots +i_j+\cdots+i_\Lambda =n\\j\in \{1,\cdots,\Lambda\}}} \operatorname{id}_{A_1^{i_1}}\otimes\cdots \otimes \delta_{A_j}^{i_j}\otimes\cdots \otimes \operatorname{id}_{A_\Lambda^{i_\Lambda}},
\end{equation}
where $\operatorname{id}$ is the identity map.
\begin{exa}{Square and Cube complexes} \label{exa: cube complexes}
    Consider the case we have three 2-term cochain complexes
\begin{center}
    \begin{tikzcd}[cells={nodes={minimum height=2em}}]
    A = (A^0\arrow[r, "\delta^0_{A}"] & A^1),\\
    B = (B^0\arrow[r, "\delta^0_{B}"] & B^1),\\
    C = (C^0\arrow[r, "\delta^0_{C}"] & C^1).
    \end{tikzcd}
\end{center}
Taking the tensor product of $A$ and $B$ returns the square double complex
\begin{center}
\begin{tikzcd}[
    cells={nodes={minimum height=2em}}, 
    column sep=6em, %
    row sep=5em    %
]
A^0\otimes B^0 \arrow[d,"id_{A^0}\otimes \delta^0_{B}"] \arrow[r, "\delta^0_{A}\otimes id_{B^0}"] & A^1\otimes B^0 \arrow[d,"id_{A^1}\otimes \delta^0_{B}"]\\
A^0\otimes B^1 \arrow[r, "\delta^0_{A}\otimes id_{B^1}"] & A^1\otimes B^1
\end{tikzcd}
\end{center}
This can be collapsed into the total complex which we write as
\begin{center}
    \begin{tikzcd}[
    cells={nodes={minimum height=2em}}, 
    column sep=10em, %
    row sep=5em    %
]
    A^0\otimes B^0 \arrow[r,"H_X^T := {\begin{bmatrix}
        \delta^0_{A}\otimes id_{B^0}\\
        id_{A^0}\otimes\delta^0_{B}
    \end{bmatrix}}"] & A^1\otimes B^0\oplus A^0\otimes B^1 \arrow[r,"{H_Z :=[id_{A^1}\otimes\delta^0_{B}\mid \delta^0_{A}\otimes id_{B^1}]}"] & A^1\otimes B^1,
    \end{tikzcd}
\end{center}
where we can choose the way in which we stack the new coboundary operators. This is a 3-term cochain complex, which defines a CSS code. 
Similarly, taking the tensor product of $A,B$ and $C$ returns the cube complex
\begin{center}
\begin{tikzpicture}[>=stealth, node distance=3cm, auto,scale = 1, font = \small]
\def\shift{1.75}
\def\nodepos{4.5}
\node (A000) at (0,0) {$A^0\otimes B^0\otimes C^0$};
\node (A100) at (\nodepos,0) {$A^1\otimes B^0\otimes C^0$};
\node (A010) at (0,\nodepos) {$A^0\otimes B^1\otimes C^0$};
\node (A110) at (\nodepos,\nodepos) {$A^1\otimes B^1\otimes C^0$};

\node (A001) at (\shift,\shift) {$A^0\otimes B^0\otimes C^1$};
\node (A101) at (\nodepos+\shift,\shift) {$A^1\otimes B^0\otimes C^1$};
\node (A011) at (\shift,\nodepos+\shift) {$A^0\otimes B^1\otimes C^1$};
\node (A111) at (\nodepos+\shift,\nodepos+\shift) {$A^1\otimes B^1\otimes C^1$};

\draw[->] (A000) -- node[below] {$\tilde{\delta}^0_{A}$} (A100);
\draw[->] (A000) -- node[left]  {$\tilde\delta^0_{B}$} (A010);
\draw[->] (A100) -- node[right] {$\tilde\delta^0_{B}$} (A110);
\draw[->] (A010) -- node[above] {$\tilde\delta^0_{A}$} (A110);

\draw[->] (A001) -- node[below] {$\tilde\delta^0_{A}$} (A101);
\draw[->] (A001) -- node[left]  {$\tilde\delta^0_{B}$} (A011);
\draw[->] (A101) -- node[right] {$\tilde\delta^0_{B}$} (A111);
\draw[->] (A011) -- node[above] {$\tilde\delta^0_{A}$} (A111);

\draw[->] (A000) -- node[left] {$\tilde\delta^0_{C}$} (A001);
\draw[->] (A100) -- node[left] {$\tilde\delta^0_{C}$} (A101);
\draw[->] (A010) -- node[left] {$\tilde\delta^0_{C}$} (A011);
\draw[->] (A110) -- node[left] {$\tilde\delta^0_{C}$} (A111);
\end{tikzpicture}
\end{center}
The notation $\tilde\delta$ refers to tensor products of coboundary operator $\delta$ with the appropriate identity maps, which we have omitted for clarity.
This can be collapsed into a 4-term total complex with coboundary operators obtained through appropriate stacking of the input coboundary operators. For examples of this, see e.g.\ \cite[Appendix B]{eberhardt2024logicaloperatorsfoldtransversalgates}, \cite[Section II]{jacob2025singleshotdecodingfaulttolerantgates}, \cite[Appendix A]{menon2025magictricyclesefficientmagic}. A CSS code can be defined by excising 3 consecutive terms of the cochain complex. We remark that, as has been pointed out in \cite[Appendix B]{eberhardt2024logicaloperatorsfoldtransversalgates}, quantum codes defined on the product of 3 cochain complexes are not symmetric as the number of sectors with the same grading is different no matter how consecutive terms are excised from the total complex. The extension to $\Lambda$-dimensional hypercube complexes from the tensor product of $\Lambda$ 2-term cochain complexes is straightforward: the resultant $(\Lambda+1)$-term total complex $E$ has terms $E^k,k\in\{0,\cdots,\Lambda\}$ as direct sums of ${\Lambda\choose k}$ terms from the hypercube complex with the same grading, where ${\Lambda\choose k}$ are binomial coefficients.
\end{exa}
\subsection{Pre-orientations and cup products} \label{sec: cup product}
On complexes derived from manifolds, the cup product is a product on suitably-oriented cochains that outputs another cochain in the complex. In this work, we specifically consider the cup product on $0$- and $1$-cochains. To generalise the product to arbitrary complexes, including ones that we can construct from tensor products of classical codes, we first have to define a notion of \emph{pre-orientation}, a generalisation of cochain orientation to a general cochain complex.
\begin{definition}[Pre-orientation] \label{defn: preorientation}
    Consider a based cochain complex $A$ with terms in the cochain complex $\{A^i\}$, basis $\{X_i^A\}$, and coboundary operators $\delta$. Let us identify $i$-cochains with their support as binary vectors. A \emph{pre-orientation} is a disjoint union of 1-coboundaries $\delta(a^0)$ such that
    \begin{equation}
        \delta(a^0) = \delta_{in}(a^0)\uplus \delta_{out}(a^0)\uplus\delta_{free}(a^0)
    \end{equation}
for every $a^0\in X_0^A$. The disjoint union $\uplus$ indicates the relation
\begin{equation}
    \delta_{\alpha}(a^0)\cap\delta_{\beta}(a^0) = \emptyset,
\end{equation}
for $\alpha,\beta\in\{in,out,free\},\alpha\neq\beta$, where $\cap$ indicates the set intersection on the vector supports.
\end{definition}
Informally, thinking of $\delta$ as a matrix acting on vectors in $A^0$, a pre-orientation is the labelling of every non-zero entry in the matrix with the subscript \emph{in}, \emph{out}, or \emph{free}. 
In complexes constructed from the cellulations of manifolds, each $a^1\in X_1$ connects two $a^0,b^0\in X_0$. Similarly, in graphs, every edge connects two vertices. In both of these cases, pre-orientation is simply an assignment of edge orientation with $\delta_{free}(a^0)=0$ for all basis elements $a^0\in X_0$. 
For example, let an oriented graph be defined by an incidence matrix with columns indexing vertices and rows indexing edges. Then every row of the incidence matrix has 2 non-zero entries. Each entry is labelled by an \emph{in} or \emph{out} denoting whether edges start from or end on the vertex, and the disjoint partitions $\delta_{in}$ and $\delta_{out}$ are simply the splitting of non-zero entries in~$\delta$ based on their labelling. %

Classical LDPC codes can be defined by a parity check matrix, where it is not necessary for rows or columns of the check matrices to have only 2 non-zero entries. Interpret as an incidence matrix, the check matrices denote \emph{hypergraphs}, where \emph{hyperedges} can join arbitrary numbers of vertices. A pre-orientation can therefore be thought of as a generalisation of edge orientation to hypergraphs; and from topological complexes to general cochain complexes. 
\begin{definition}[Cup product on 0- and 1-cochains] \label{defn: cup product}
    Let 1-cochains (edges) be indexed by superscripts $^1$ and 0-cochains (vertices) be indexed by superscripts $^0$, and a basis $X_i$ be defined for $i$-cochains. Then the \emph{cup product} between $x^i\in X_i$ and $y^j\in X_j$, $i,j\in \{0,1\}$ is defined by
    \begin{subequations}
        \begin{align}
            x^1\cup x^0 &= x^1,x^1 \in \delta_{in}(x^0),\\
            x^0\cup x^1 &= x^1,x^1 \in \delta_{out}(x^0),\\
            x^0\cup y^0 &= x^0,x^0 =y^0
        \end{align}
    \end{subequations}
and zero otherwise. The product extends by linearity to all $i$-cochains.

Cup products on tensor product cochains can similarly be defined by distributing the cup product across the tensor product. Explicitly, for cochains $a,b,c,d$, we have that 
    \begin{equation}
            a\otimes b \cup c\otimes d := (a\cup c) \otimes (b\cup d).
    \end{equation}
\end{definition}
When there are more than 2 arguments, the cup product can be non-associative, and depends on the order the cup products are applied, defined by brackets. Consider, for example, the case where both input cochains are 1-cochains from a tensor product cochain complex. In \cite{jacob2025singleshotdecodingfaulttolerantgates}, the cup product is done with brackets originating left-to-right, regardless of where the 1-cochain argument sits in the tensor product. In \cite{menon2025magictricyclesefficientmagic,breuckmann2024cupsgatesicohomology}, the brackets originate \emph{from} the 1-cochain, always applied between 1-cochains and 0-cochains, ``inside-out". We define the former cup product as the \emph{non-associative} cup product, and the latter cup product the \emph{symmetric} cup product according to terminology of \cite{menon2025magictricyclesefficientmagic}. We can also, for example, define a third way non-associativity is handled, by originating brackets ``outside-in", where the cup product on the 1-cochain is handled last. We illustrate the differences with an example on the cup product of 3 cochains.
\begin{exa}{Cup products on 3-term tensor product cochains} \label{exa: cup products}
    Consider a cup product of the form
    \begin{equation}
        a^1\otimes b^0\otimes c^0 \cup d^0\otimes e^1\otimes f^0 \cup g^0\otimes h^0\otimes i^1
    \end{equation}
    where the superscript indicates whether it is a 0- or 1-cochain; and in this case are all basis elements. This form of cup product is encountered in the 3-copy-cup gate. In \cite{jacob2025singleshotdecodingfaulttolerantgates}, the cup product brackets are applied as
    \begin{subequations}
        \begin{align}
            &a^1\otimes b^0\otimes c^0 \cup d^0\otimes e^1\otimes f^0 \cup g^0\otimes h^0\otimes i^1\\
            &= ((a^1\otimes b^0\otimes c^0 \cup d^0\otimes e^1\otimes f^0) \cup g^0\otimes h^0\otimes i^1)
        \end{align}
    \end{subequations}
    where brackets always originate from the left. Then this has the simplification
    \begin{subequations}
        \begin{align}
            &((a^1\otimes b^0\otimes c^0 \cup d^0\otimes e^1\otimes f^0) \cup g^0\otimes h^0\otimes i^1)\\
            &= ((a^1\cup d^0) \cup g^0)\otimes ((b^0\cup e^1)\cup h^0) \otimes ((c^0\cup f^0) \cup i^1)\\
            &= a^1\otimes e^1\otimes i^1,\\
            a^1&\in \delta_{in}(d^0)\cap  \delta_{in}(g^0),\\
            e^1 &\in \delta_{out}(b^0)\cap \delta_{in}(h^0),\\
            i^1&\in \delta_{out}(c^0\cap f^0),
        \end{align}
    \end{subequations}
    following the rules of the cup product set out in \cref{defn: cup product}. The notation $\cap$ indicates the set intersection of cochains, which are in turn identified with their support as binary vectors as defined in \cref{defn: preorientation}. We write the condition on the $i^1$ term suggestively to highlight how the different cup products lead to different conditions on pre-orientation in \cref{defn: master} later on. On the other hand, doing the calculation with the symmetric cup product defined in \cite{menon2025magictricyclesefficientmagic,breuckmann2024cupsgatesicohomology} returns
    \begin{subequations}
        \begin{align}
            &a^1\otimes b^0\otimes c^0 \cup d^0\otimes e^1\otimes f^0 \cup g^0\otimes h^0\otimes i^1\\
            &= ((a^1\cup d^0)\cup g^0)\otimes (b^0\cup e^1\cup h^0)\otimes (c^0\cup (f^0\cup i^1)) \\
            &= a^1\otimes e^1\otimes i^1,\\
            a^1&\in \delta_{in}(d^0)\cap \delta_{in}(g^0),\\
            e^1 &\in \delta_{out}(b^0)\cap \delta_{in}(h^0),\\
            i^1&\in \delta_{out}(c^0)\cap \delta_{out}(f^0).
        \end{align}
    \end{subequations}
    Finally, we could also apply the cup product ``outside-in" with cup products on the 1-cochain applied last, which returns
    \begin{subequations}
        \begin{align}
            &a^1\otimes b^0\otimes c^0 \cup d^0\otimes e^1\otimes f^0 \cup g^0\otimes h^0\otimes i^1\\
            &= (a^1\cup (d^0\cup g^0))\otimes (b^0\cup e^1\cup h^0)\otimes ((c^0\cup f^0)\cup i^1) \\
            &= a^1\otimes e^1\otimes i^1,\\
            a^1&\in \delta_{in}(d^0\cap g^0),\\
            e^1 &\in \delta_{out}(b^0)\cap \delta_{in}(h^0),\\
            i^1&\in \delta_{out}(c^0\cap f^0).
        \end{align}
    \end{subequations}
    The order of all cup products in the second argument of the tensor product does not matter. The difference between the symmetric and non-associative cup products presents itself in conditions for the final $i^1$ term of the tensor product.
\end{exa}
\subsection{Integrals and the integrated Leibniz rule}
We now move on to the definition of integrals and the integrated Leibniz rule. The integrated Leibniz rule leads to the conditions that must be satisfied by the classical codes in order for the tensor product quantum code to have the copy-cup gate. Importantly, we point out in this section that the restriction to even check weight codes, as in the previous works of \cite{breuckmann2024cupsgatesicohomology,jacob2025singleshotdecodingfaulttolerantgates,menon2025magictricyclesefficientmagic} is in fact unnecessary. This lets us analyse a larger family of quantum codes, including the bivariate bicycle codes \cite{bravyi2023highthreshold} which are constructed as the balanced product of weight 3 abelian group algebra codes.

\begin{definition}[$n$-dimensional Integral] \label{integral}
An \emph{$n$-dimensional integral} on $n$-cochains from a complex $A$ is defined as
\begin{subequations}
    \begin{align}
    \int^A_n:A^n&\to \mathbb{F}_2,\\
    \int^A_n b &= 0
\end{align}
\end{subequations}
where $A^n$ indicates the space of $n$-cochains, and such that the integral maps any coboundary $b\in \operatorname{im} \delta^{n-1}$ to 0.
\end{definition}
Note that there is a clash of notation with the definite integral: the superscript indicates the cochain complex $A$, while the subscript $n$ indicates that the integral is a function on $n$-cochains. In 2-term cochain complexes where
\begin{center}
\begin{tikzcd}[cells={nodes={minimum height=2em}}]
A = A^0 \arrow[r, "\delta^0"] & A^1,
\end{tikzcd}
\end{center}
a candidate for an integral on 1-cochains is the Hamming weight of the cochain modulo 2, i.e.
\begin{subequations}\label{eqn: 1d integral}
   \begin{align} 
    \int^A_1:A^1&\to \mathbb{F}_2,\\
    \int^A_1 b &\mapsto |b| \mod 2,
\end{align}
\end{subequations}
This was considered in \cite{breuckmann2024cupsgatesicohomology}, and is what we will use as the integral for the rest of the paper.

Interpreting the 2-term complex as a classical code, if the check weight of the code is even, i.e. column weight of~$\delta^0$ is even, then all coboundaries in $\operatorname{im}\delta^0$ will have 0 integral. Let us remark that, by this definition, the integral on classical codes with odd check weight cannot be 0.

We can further define structure on integrals in the case of tensor product cochain complexes.
\begin{definition}[Integrals on tensor product cochain complexes] \label{defn: tensor product cochain}
Let $A_1$ and $A_2$ be two cochain complexes, where the terms $A_1^i,A_2^j$ in the complexes are equipped with integrals $\int^{A_1}_i,\int^{A_2}_j$ of dimensions $i$ and $j$ respectively. Then their tensor product cochain complex is equipped with the integral
\begin{subequations}
   \begin{align}
    \int_{i+j}^{A_1\otimes A_2}:A_1^i\otimes A_2^j&\to \mathbb{F}_2,\\
    \int_{i+j}^{A_1\otimes A_2} a_1\otimes a_2&\mapsto \int_i^{A_1}a_1\int_j^{A_2}a_2.
\end{align}
\end{subequations}
\end{definition}
Similar to the treatment of cup products, we distribute the integral across terms of the tensor product cochain. We can use \cref{defn: tensor product cochain} in particular with a $\Lambda$-fold tensor product cochain complex, where all tensor product arguments are 1-cochains.

\begin{definition} \label{lambda_integral}
Let a $(\Lambda+1)$-term cochain complex $A$ be defined by the tensor product of $\Lambda$ 2-term cochain complexes~$\{A_i\}$, $i\in \{1,\dotsc,\Lambda \}$. Then the $\Lambda$-dimensional integral on $\Lambda$-cochains is defined as
\begin{subequations}
   \begin{align}
    \int_\Lambda:\bigotimes_{i=1}^{\Lambda} A^{1}_i&\to \mathbb{F}_2,\\
    \int_\Lambda a_1\otimes \cdots \otimes a_\Lambda &\mapsto \int_1a_1\int_1a_2\cdots\int_1a_\Lambda.
\end{align}
\end{subequations}
\end{definition}
The superscript is omitted on individual integrals, and we take all 1-dimensional integrals to be the Hamming weight modulo 2, just as in \cref{eqn: 1d integral}. The higher-dimensional integrals and cup products allow us to analyse longer-term cochain complexes, and in particular the tensor product cochain complexes. We can now define the copy-cup gate on $\Lambda$ copies of the $(\Lambda+1)$-term cochain complex.
\begin{definition}[$\Lambda$-copy-cup gate] \label{defn: copy-cup gate}
    Consider $\Lambda$ copies of a $(\Lambda+1)$-term cochain complex constructed from the tensor product of $\Lambda$ 2-term cochain complexes. The $i^{th}$ copy of the $(\Lambda+1)$-term cochain complex has a basis of 1-cochains $\underline{X}_1^{(i)}$, and the union of the complexes is the quantum code with qubits identified with $\underline{X}=\underline{X}_1^{(1)}\uplus\cdots\uplus \underline{X}_1^{(\Lambda)}$. The \emph{$\Lambda$-copy-cup gate} is the logical gate implemented by the circuit
    \begin{equation} \label{copy-cup gate}
        \prod_{\substack{\{x_i\in \underline{X}_1^{(i)}\}\\
        i\in\{1,\cdots,\Lambda\}}}\left(C^{\Lambda-1}Z_{x_1\cdots x_\Lambda}\right)^{\int_\Lambda x_1\cup \cdots \cup x_\Lambda},
    \end{equation}
    where $\int_\Lambda$ is as defined in \cref{lambda_integral}. The circuit involves gates at the $\Lambda^{th}$ level of the Clifford hierarchy.
\end{definition}
We freely identify the gate and the circuit described in \cref{defn: copy-cup gate}. To show that the copy-cup gate acts as a logic gate on the quantum code, we must show that the integral
\begin{equation}
    \int_\Lambda x_1\cup\cdots \cup x_\Lambda \label{eqn: cohom inv}
\end{equation}
in \cref{copy-cup gate} is a cohomology invariant as defined in \cref{sec: cohom invariants}.
This requires the definition of the integrated Leibniz rule; and if the integrated Leibniz rule is satisfied, the copy-cup gate acts as an operation in cohomology.
\begin{definition}[Integrated Leibniz rule]
A 2-term cochain complex with coboundary operator $\delta$ and basis $X_0$ for 0-cochains satisfies an integrated Leibniz rule if
    \begin{equation} 
    \sum_{i=1}^{\Lambda} \int _1 a_1\cup \cdots \cup \delta(a_i)\cup \cdots a_\Lambda =0 \mod 2
\label{eqn: integrated leibniz} \end{equation}
for all $a_i\in X_0, i\in \{1,\dotsc,\Lambda\}$.
\end{definition}
\begin{lemma}[Cohomology invariant] \label{lem: cohom invariance} The tensor product of $\Lambda$ 2-term cochain complexes that satisfy the integrated Leibniz rule up to $\Lambda$ returns a cochain complex that has \cref{eqn: cohom inv} as a cohomology-invariant operation. As a quantum code with qubits defined on a basis $X_1$ of 1-cochains, the cochain complex will be equipped with the $\Lambda$-copy cup gate as defined in \cref{defn: copy-cup gate}.
\end{lemma}
We direct readers interested in the proof to \cite[Lemma 5.1 and Theorem 6.1]{breuckmann2024cupsgatesicohomology}. The integrated Leibniz rule leads to the copy-cup gate respecting coboundaries in $\Lambda$ copies of the $\Lambda$-fold tensor product of 2-term cochain complexes, and therefore acting as a cohomology operation.

Understanding the integrated Leibniz rule lets us consider classical codes with odd check weights in the construction of quantum codes with 2- and 3-copy-cup gates, which have not been considered in previous works \cite{breuckmann2024cupsgatesicohomology,menon2025magictricyclesefficientmagic,jacob2025singleshotdecodingfaulttolerantgates}. When $\Lambda=1$, we recover the integral $\int_1$ in \cref{eqn: 1d integral} from the integrated Leibniz rule on a single copy of a classical code. This definition implies that classical codes with check weights of odd order, such as odd-term group algebra codes, do not satisfy the 1-copy-cup gate, since coboundaries (columns in the coboundary operator, resp. rows in the code parity check matrix) are no longer mapped to 0 modulo 2. 
Indeed, in \cite{breuckmann2024cupsgatesicohomology}, this was the reason why only classical codes with even check weight were considered when constructing the tensor product quantum codes.

However, in the proof of \cref{lem: cohom invariance}, particularly in \cite[Lemma 5.1]{breuckmann2024cupsgatesicohomology}, we only require that the integrated Leibniz rule is required to be satisfied up to $\Lambda$ for the $\Lambda$-copy-cup gate to be a cohomology invariant, and conditions for different $\Lambda$ need not be satisfied simultaneously. We could therefore have a classical code of odd weight that satisfies conditions for 2- and 3-copy-cup gates but not the 1-copy cup gate, which does not correspond to a logical quantum gate. Indeed, as we will see in later sections, conditions on the codes to satisfy the 2- and 3-copy-cup gates can be treated independently. There is therefore no need to restrict ourselves to discussion of codes with even check weight. Immediately, this lets us consider copy-cup gates on group algebra codes of odd order, including bivariate bicycle codes from \cite{bravyi2023highthreshold}.

\subsection{Conditions on codes from the integrated Leibniz rule}
We have seen that the integrated Leibniz rule needs to be satisfied by classical codes for the tensor product quantum code to have the copy-cup gate as a logic gate. We now show how the integrated Leibniz rule results in conditions on pre-orientation, and therefore restrictions on the classical code parity check matrices.

\begin{definition}[Master equations]
\label{defn: master}
The integrated Leibniz rule is satisfied by a classical code if, for all $a_i\in X_0$, it satisfies
    \begin{equation} \label{eqn: master}
        \sum_{j=1}^{\Lambda}\left\vert\delta_{out}\left(\bigcap_{i=1}^{j-1}(a_i)\right)\cap\delta (a_j) \cap \left(\bigcap_{k=j+1}^\Lambda \delta_{in}(a_k)\right) \right\vert = 0 \mod 2
    \end{equation}
for the non-associative cup product, where cup products are applied left-to-right,
    \begin{equation} \label{eqn: master2}
        \sum_{j=1}^{\Lambda}\left\vert\left(\bigcap_{i=1}^{j-1}\delta_{out}(a_i)\right)\cap\delta (a_j) \cap \left(\bigcap_{k=j+1}^\Lambda \delta_{in}(a_k)\right) \right\vert = 0 \mod 2
    \end{equation}
for the symmetric cup product defined in \cite{menon2025magictricyclesefficientmagic,breuckmann2024cupsgatesicohomology}, where brackets are applied to 1-cochains first, and 
    \begin{equation} \label{eqn: master3}
        \sum_{j=1}^{\Lambda}
        \left\vert
        \delta_{out}\left(\bigcap_{i=1}^{j-1}(a_i)\right)
        \cap\delta (a_j) 
        \cap  \delta_{in}\left(\bigcap_{k=j+1}^\Lambda(a_k)\right) 
        \right\vert = 0 \mod 2
    \end{equation}
for an outside-in cup product where cup products are applied to 1-cochains last.
\end{definition}
We remind the reader that, as in \cref{defn: preorientation}, $\cap$ refers to a set intersection, and we freely identify vectors with their supports. The definitions differ solely in how non-associativity in the cup product is handled, and we refer the reader to \cref{defn: cup product} and the discussion following it for an explicit example that illustrates the different cup products. 
Satisfaction of the integrated Leibniz rule follows from these equations by substituting them into \cref{eqn: integrated leibniz}.
From these equations, we will be able to derive all possible constraints on the classical codes.
We therefore label them as ``master equations". In particular, \cref{eqn: master2} was presented in \cite[Proposition 5.2]{breuckmann2024cupsgatesicohomology}, while \cref{eqn: master} has been considered in \cite{jacob2025singleshotdecodingfaulttolerantgates}.
For the rest of the paper, we will focus on the non-associative and symmetric cup products.

Generally, non-associativity in the cup product leads to different constraints on codes. However, when the pre-orientation satisfies the \emph{non-overlapping bits} condition, the cup product is associative, leading to the same set of code constraints from all equations in \cref{defn: master}.
\begin{definition}[Non-overlapping bits, Definition 5.3 in \cite{breuckmann2024cupsgatesicohomology}]
\label{nonoverlap}
A pre-orientation is said to satisfy the \emph{non-overlapping bits} condition if 
\begin{subequations}
\begin{align}
    \delta_{in}(a_1)\cap\delta_{in}(a_2)&=\emptyset,\qquad a_1\neq a_2,\\
    \delta_{out}(a_1)\cap\delta_{out}(a_2)&=\emptyset,\qquad a_1\neq a_2.
\end{align}
\end{subequations}
\end{definition}
The non-overlapping bits condition leads to associativity in the cup product. For example, it is immediately satisfied in group algebra codes that have single terms in $\delta_{in}$ and $\delta_{out}$; and also for directed graphs with the edges leading into/out of vertices being labelled as $\delta_{in}$ and $\delta_{out}$ terms respectively.

At this stage, let us simplify notation of the multiple intersections in the equations. We assign the $\delta_{in},\delta_{out},\delta_{free}$ terms to indices $i,o,f$ respectively, and write the intersections as indices according to their positions. For example,  $|\delta(a_1)\cap\delta_{in}(a_2)|$ is written
\begin{subequations}
\begin{align}
    |\delta(a_1)\cap\delta_{in}(a_2)| &=: (i+o+f)i\\
    &= ii+oi+fi\\
    &=|\delta_{in}(a_1)\cap\delta_{in}(a_2)| \\\notag&+ |\delta_{out}(a_1)\cap\delta_{in}(a_2)| \\\notag&+ |\delta_{free}(a_1)\cap\delta_{in}(a_2)|.
\end{align}
\end{subequations}
For the remainder of this section, we will refer to the intersections using this notation, and omit the $a_i$ which are understood to be from a basis $X_0$.
\begin{corollary}
When the non-overlapping bits condition of \cref{nonoverlap} is satisfied, we have that
\begin{subequations} \label{eqn: repeated numbers}
\begin{align}
    ii&=0,\qquad a_1\neq a_2,\\
    oo&=0, \qquad a_1\neq a_2.
\end{align}
\end{subequations}
\end{corollary}

\cref{defn: master,nonoverlap} let us determine the necessary conditions on the pre-orientation to have a valid cup product.

\begin{prop}[2-copy-cup gate conditions]\label{prop ryan ver}
Let $\Lambda=2$. The conditions, defined over modulo 2, on pre-orientation of a classical code for the 2-copy-cup gate are given as 
\begin{subequations} \label{eqn: CZ preorientation conditions}
\begin{align}
    i+o&=0, \qquad a_1=a_2,\label{eqn: parity CZ non associative}\\
    ii+fi+oo+of&=0, \qquad a_1\neq a_2.\label{eqn: CZ non associative 2}
\end{align}
\end{subequations}
When the no-overlapping bits condition is satisfied, the set of equations simplify to
\begin{subequations}
\begin{align}
    i+o&=0, \qquad a_1=a_2,\\
    of+fi &=0, \qquad a_1\neq a_2. \label{eqn: CZ associative}
\end{align}
\end{subequations}
\end{prop}
\begin{proof}
When $\Lambda=2$ the symmetric, non-associative, and outside-in cup products return the same conditions from the master equations of \cref{defn: master}; that is,
\begin{equation}
    (i+o+f)i + o(i+o+f) = 0\mod2.
\end{equation}
Setting $a_1=a_2$ returns us the condition \cref{eqn: parity CZ non associative} by noting that other terms cancel since $\delta$ is a disjoint union of coboundaries. Similarly, setting $a_1\neq a_2$ and cancelling terms returns us \cref{eqn: CZ non associative 2}. When the no-overlapping bits condition is satisfied, terms with repeated indices cancel. This is because they correspond to intersections of the form in \cref{eqn: repeated numbers}; removing the terms returns \cref{eqn: CZ associative}.
\end{proof}
The cup product on 2 arguments is associative, so the conditions on the 2-copy-cup gate is the same for the symmetric, non-associative, and outside-in cup products.
We can reason conditions for the 3-copy-cup gate in a similar fashion. However, in this case, due to non-associativity, the master equations result in different sets of conditions depending on how the cup product is applied.
\begin{prop}[3-copy-cup gate conditions] \label{prop: lambda 3 conditions}
Let $\Lambda=3$. The conditions, defined over modulo 2, on pre-orientation of a classical code for the 3-copy-cup gate using the non-associative cup product are given as
\begin{subequations}
\label{eqn: non associative cup product}
\begin{align}
    i+o&=0,\label{eqn: parity CCZ non associative}\\
    oo+of&=0, \label{eqn: 22+23=0}\\
    fi&=0,\label{eqn:singular1}\\
    ii&=0,\label{eqn:singular2}\\
    iii+fii+ooi+ofi&=0. \label{eqn: 111331 blabla bla}
\end{align}
\end{subequations}
For the symmetric cup product, the conditions, defined over modulo 2, are
\begin{subequations}
\label{eqn: symmetric non associative cup product}
\begin{align}
    i+o&=0, \label{symmetric CCZ parity}\\
    ii+oo&=0, \label{eqn: 11+22 symmetric} \\
    of&=0, \label{eqn: 23=0 symmetric}\\
    fi&=0, \label{eqn: 31=0 symmetric}\\
    iii+fii+ofi+ooo+oof&=0. \label{eqn: 111 blablabla symmetric}
\end{align}
\end{subequations}
If $\delta_{free}=\emptyset$ then the equations simplify to
\begin{subequations} \label{eqn: menons free}
\begin{align}
    i+o&=0, \\
    ii+oo&=0, \\
    iii+ooo &=0,\\
    \delta_{free}:=\emptyset. \notag
\end{align}
\end{subequations}
For the outside-in cup product, the conditions, defined over modulo 2, are given as
\begin{subequations}
\label{eqn: outside in cup product}
\begin{align}
    i+o&=0,\\
    oo+of&=0, \label{singulara}\\
    ii+fi&=0, \label{singularb}\\
    oii+ooi+ofi&=0. \label{fat eqn outside in}
\end{align}
\end{subequations}
In all cases, if the no-overlapping bits condition is satisfied, the conditions simplify to
\begin{subequations} \label{eqn: associative CCZ conditions}
\begin{align}
    i+o&=0,\\
    of&=0, \label{eqn: associative cup product singular 1}\\
    fi&=0, \label{eqn: associative cup product singular 2}\\
    ofi&=0.
\end{align}
\end{subequations}
\end{prop}
\begin{proof}
Let $\Lambda=3$. From \cref{eqn: master} we get the equation
\begin{equation}
    (i+o+f)ii + o(i+o+f)i+\delta_{a_1,a_2}oo(i+o+f)=0\mod 2
\end{equation}
for the non-associative cup product.
Note that the $\delta_{a_1,a_2}$ term refers to a delta function rather than a coboundary operator, and is non-zero only when $a_1=a_2$. The delta function comes from the $\delta_{out}(a_1\cap a_2)$ term of \cref{eqn: master} when $j=2$. Canceling terms and setting $a_1,a_2,a_3$ equalities, we obtain the set of equations over modulo 2
\begin{subequations}
\begin{align}
    i+o&=0,\qquad a_1=a_2=a_3\\
    oo+of&=0, \qquad a_1= a_2\neq a_3\\
    ii+fi&=0,\qquad a_1\neq a_2=a_3 \label{eqn: unsimplified}\\
    ii&=0,\qquad a_1=a_3\neq a_2 \label{eqn:singular3}\\
    iii+fii+ooi+ofi&=0 ,\qquad a_1\neq a_2\neq a_3\neq a_1,
\end{align}
\end{subequations}
where for this first set of equations we explicitly list out how to obtain each equation by setting equalities on $a_i$. Using \cref{eqn:singular3} we can simplify \cref{eqn: unsimplified}; omitting the $a_i$ equalities, we obtain the set of equations in \cref{eqn: non associative cup product}. Then, by considering the no-overlapping bits condition, all terms with repeated indices of the form $ii$, $oo$ cancel, and we get \cref{eqn: associative CCZ conditions}.

Similarly, from \cref{eqn: master2} for the symmetric cup product, we start with the equation
\begin{equation}
    (i+o+f)ii + o(i+o+f)i+oo(i+o+f)=0\mod2.
\end{equation}
Cancelling terms and setting $a_1,a_2,a_3$ equalities, we obtain the set of equations over modulo 2
\begin{subequations} \label{eqn: symmetric cup product CCZ}
\begin{align}
    i+o&=0,\\
    ii+oo+of&=0, \label{eqn: redundancy2}\\
    ii+fi+oo&=0,\label{eqn: redundancy3}\\
    ii+oo&=0, \label{eqn: redundancy}\\
    iii+fii+ofi+ooo+oof&=0.
\end{align}
\end{subequations}
Again we can use \cref{eqn: redundancy} to simplify \cref{eqn: redundancy2,eqn: redundancy3} to obtain the set of equations defined in \cref{eqn: symmetric non associative cup product}; and similar cancellation of terms when assuming the no-overlapping bits condition returns \cref{eqn: associative CCZ conditions}.

Finally, for the outside-in cup product, we start with \cref{eqn: master3} and set inequalities and cancel terms in the same way to obtain \cref{eqn: outside in cup product} and \cref{eqn: associative CCZ conditions}.
\end{proof}
\begin{remark}
    The conditions in \cref{eqn: non associative cup product} is the same set of equations as in \cite[Equations 53-57]{jacob2025singleshotdecodingfaulttolerantgates}. Similarly, the conditions of \cref{eqn: symmetric non associative cup product} and \cref{eqn: menons free} are given in \cite[Equations 46-50, Equations 52-54]{menon2025magictricyclesefficientmagic} respectively.
\end{remark}
In all cup products, when the no-overlapping bits condition (\cref{nonoverlap}) is satisfied, we obtain the same set of equations as \cref{eqn: associative CCZ conditions}. This makes sense since the different cup products are defined to handle non-associativity when there are 3 or more terms in the cup product. When the no-overlapping bits condition is fulfilled, the cup product is associative and we return to the same set of conditions. 

We also make a remark on conditions such as those in \cref{eqn: 31=0 symmetric,eqn: 23=0 symmetric,eqn:singular1,singulara,singularb}. These equations involve intersections terms between different partitions, such that a partition can be ``factored out". We call these \emph{singular conditions}, and we will see that singular conditions lead to unsatisfiable constraints in pre-orientations of group algebra codes. 
\begin{definition}[Singular conditions] \label{defn: singular conditions}
    Consider the systems of equations presented in \cref{prop ryan ver} and \cref{prop: lambda 3 conditions}, which are derived from the master equations in \cref{defn: master}. The systems of equations involve intersection terms between partitions $\delta_{in},\delta_{out}, \delta_{free}$ of the coboundary operators, acting on $a_1,a_2\in X_0$ if it is a double intersection, and $a_1,a_2,a_3\in X_0$ if it is a triple intersection. 
    A \emph{singular condition} is an equation involving double or triple intersection terms, such that terms in the equation can be written as
    \begin{subequations}
    \begin{align}
        |\delta_{\alpha}(a_1)\cap\delta_{\beta_1}(a_2)|
        +\cdots
        +|\delta_{\alpha}(a_1)\cap\delta_{\beta_m}(a_2)|
        &= |\delta_{\alpha}(a_1)\cap(\delta_{\beta_1}+\cdots+\delta_{\beta_{m}})(a_2)|, 
        \text{ or } \\
        |\delta_{\beta_1}(a_1)\cap\delta_{\alpha}(a_2)|+\cdots+
        |\delta_{\beta_m}(a_1)\cap\delta_{\alpha}(a_2)|
        &=|(\delta_{\beta_1}+\cdots+\delta_{\beta_{m}})(a_1)\cap \delta_{\alpha}(a_2)|
        \end{align}
    \end{subequations}
    if the equation involves double intersection terms, or    \begin{subequations}
    \begin{align}
        &|\delta_{\alpha}(a_1)\cap\delta_{\beta_1}(a_2)\cap\delta_{\gamma_1}(a_3)|
        +\cdots
        +|\delta_{\alpha}(a_1)\cap\delta_{\beta_m}(a_2)\cap\delta_{\gamma_m}(a_3)|, 
        \text{ or } \\
        &|\delta_{\beta_1}(a_1)\cap\delta_{\gamma_1}(a_2)\cap\delta_{\alpha}(a_3)|+\cdots+
        |\delta_{\beta_m}(a_1)\cap\delta_{\gamma_m}(a_2)\cap\delta_{\alpha}(a_3)|
        \end{align}
    \end{subequations}  
    if the equation involves triple intersection terms, for $\alpha,\beta_i,\gamma_i\in \{in,out,free\}$ and $i\in\{1,\cdots,m\}$. We exclude equations comprising a single term of the form
    \begin{subequations}
        \begin{align}
            &|\delta_{\alpha}(a_1)\cap\delta_{\alpha}(a_2)| \text{ or}\\
            &|\delta_{\beta}(a_1)\cap\delta_{\alpha}(a_2)\cap\delta_{\alpha}(a_3)|\text{ or}\\
            &|\delta_{\alpha}(a_1)\cap\delta_{\alpha}(a_2)\cap\delta_{\beta}(a_3)|\text{ or}\\
            &|\delta_{\alpha}(a_1)\cap\delta_{\beta}(a_2)\cap\delta_{\alpha}(a_3)|
        \end{align}
    \end{subequations}
    for $\alpha,\beta\in \{in,out,free\}$ from this definition.
\end{definition}

\begin{exa}{Singular conditions}
    Equations with single terms involving different coboundary partitions such as 
    \begin{align*}
        of&=|\delta_{out}(a_1)\cap\delta_{free}(a_2)|\\
        &=0 \mod 2
    \end{align*} from \cref{eqn: 23=0 symmetric} can be factored trivially, and are singular conditions. 

    Equations with multiple terms such as
    \begin{align*}
        oo+of&=|\delta_{out}(a_1)\cap\delta_{out}(a_2)| + |\delta_{out}(a_1)\cap\delta_{free}(a_2)|\\
        &=|\delta_{out}(a_1)\cap(\delta_{out}(a_2)+\delta_{free}(a_2))|\\
        &=o(o+f)\\
        &=0 \mod 2
    \end{align*}
    from \cref{singulara}, where $\delta_{out}(a_1)$ can be factored out, are singular conditions. 
    
    On the other hand, equations with intersection terms involving single coboundary partitions such as
    \begin{align*}
        ii&=|\delta_{in}(a_1)\cap\delta_{in}(a_2)|\\
        &=0\mod2
    \end{align*}
    from \cref{eqn:singular2} and \cref{eqn:singular3} are \textit{not} singular conditions since they consist of intersection terms of a \emph{single} coboundary partition. The reason we make the distinction will be clear in \cref{lemma: expansion}. There, we consider the case where the factored coboundary partition has a single element, which leads to unsatisfiable conditions on group algebra codes. If the equation contains single intersection terms from the same coboundary partition, then when the partition has a single element, the equation is immediately satisfied and does not lead to contradiction.
\end{exa}
\section{Copy-cup gates on 3-term Group Algebra Codes} \label{sec: application}
In this section, we focus on applying the conditions derived for pre-orientation on 3-term group algebra codes. 
We recall from \cref{sec: group algebra} that a group algebra code can be defined by a 2-term cochain complex with coboundary operator $\delta\in\mathbb{F}_2[G]$; and since we identify group elements with their regular representation, which are permutation matrices, we can consider the Hamming weight of \ $|\delta|$ to be the check weight of the code. Moreover, since $|\delta g| = |\delta|$ for $g\in G$, we can also define the Hamming weight of the pre-orientations $|\delta_{in}|,|\delta_{out}|,|\delta_{free}|$ as the number of group elements in each partition, without reference to multiplication with $a_i\in G$ in \cref{prop ryan ver} and \cref{prop: lambda 3 conditions}.

\subsection{Conditions for 2-copy-cup gates}
\begin{theorem} \label{thm:2lambda}
    Consider a 3-term group algebra code with $\delta=g_1^{}+g_2^{}+g_3^{}\in \mathbb{F}_2[G]$ following the definitions in \cref{sec: group algebra}. The code has a valid pre-orientation that satisfies conditions for the 2-copy-cup gate as defined in \cref{prop ryan ver} that is dependent on group elements assigned to partitions $\delta_{in},\delta_{out}$, and $\delta_{free}$. The pre-orientation satisfies the conditions
    \begin{subequations}
        \begin{align}
            g_3^{-1}g_2 &= g_1^{-1}g_3,\\
            \delta_{in}&:=g_1,\notag\\
            \delta_{out}&:=g_2,\notag\\
            \delta_{free}&:=g_3.\notag
        \end{align}
    \end{subequations}
\end{theorem}
\begin{proof}
    For a non-trivial pre-orientation we require $|\delta_{in}|,|\delta_{out}|\neq 0$ so that the copy-cup gate as defined in \cref{copy-cup gate} is not trivial. Then, by \cref{eqn: parity CZ non associative} we have that $|\delta_{in}|=|\delta_{out}|=1$.
    We define an \emph{assignment} as the number of terms in each partition. Here, having the assignment $(|\delta_{in}|,|\delta_{out}|,|\delta_{free}|)=(1,1,1)$ immediately satisfies the no-overlapping bits condition. Then the main condition that needs to be satisfied is \cref{eqn: CZ associative}, that is
\begin{equation}
    |\delta_{in}(a_1)\cap \delta_{free}(a_2)|+|\delta_{free}(a_1)\cap \delta_{out}(a_2)|=0\mod 2
    \end{equation}
for $a_i\in G, i\in \{1,2,3\}$. To satisfy this condition, we expand the intersections, then pair terms such that if $a_1,a_2$ are chosen so that either term in the pair is non-zero, the other term will be non-zero as well. The expanded equation looks like
\begin{equation}
    |g_1a_1\cap g_3a_2| + |g_3a_1\cap g_2a_2| = 0\mod 2.
\end{equation}
We remind the reader that $g_i,a_i$ are treated as binary vectors freely identified with their support.
When $|g_1a_1\cap g_3a_2|=1$, we have that $g_1a_1=g_3a_2$. Then it follows that $a_1=g_1^{-1}g_3a_2$. Similarly, when $|g_3a_1\cap g_2a_2|=1$, we have that $a_1=g_3^{-1}g_2a_2$. Putting the conditions together, we have 
\begin{equation} \label{eqn: commutation condition for lambda 2}
    g_3^{-1}g_2 = g_1^{-1}g_3.
\end{equation}
\end{proof}
We see that satisfying pre-orientation amounts to pairing terms in the system of equations given by \cref{eqn: CZ preorientation conditions} after expanding all terms in the intersections. Pairing terms ensures that the integrated Leibniz rule is satisfied, and each pair effectively returns a condition on the group algebra code. This principle will let us approach generalisation of the pre-orientation for copy-cup gates to group algebra codes with higher weight as a combinatorial problem.

Let us make some statements on valid pre-orientations before going to specific examples of bivariate bicycle codes.
\begin{corollary} \label{cor: cyclic shifts}
    For a 3-term group algebra code with~$\delta=g_1^{}+g_2^{}+g_3^{}\in \mathbb{F}_2[G]$, and that satisfies the condition in \cref{thm:2lambda}, multiplication of $\delta$ by single elements in the group will not affect the pre-orientation.
\end{corollary}
\begin{proof}
    From \cref{eqn: commutation condition for lambda 2}, we have that $g_2 = g_3g_1^{-1}g_3$. Now consider an arbitrary group element $g_4$ that we left multiply on all the group elements to get a new group algebra code defined by $g_4g_1+g_4g_2+g_4g_3$. Then the new RHS of the equation becomes
    \begin{equation}
        g_4g_3(g_1^{-1}g_4^{-1})g_4g_3 = g_4g_3g_1^{-1}g_3
    \end{equation}
    which is exactly the new LHS. A similar cancellation will occur if we had right-multiplied instead. Therefore, if $\delta$ did not satisfy pre-orientation, multiplying the $\delta$ by single group elements will not lead to valid pre-orientation. 
\end{proof}
\begin{corollary}
    The examples of bivariate bicycle codes in \cite{bravyi2023highthreshold} do not have a pre-orientation satisfying the 2-copy-cup gate.
\end{corollary}
\begin{proof}
    The bivariate bicycle codes are a special case of \cref{thm:2lambda} where the groups are abelian. Upon inspection, none of the check polynomials presented in \cite{bravyi2023highthreshold} have the required form to satisfy \cref{thm:2lambda}; and by \cref{cor: cyclic shifts}, multiplication by monomials cannot lead to a valid pre-orientation.
\end{proof}
\subsection{Conditions for 3-copy-cup gates}
\begin{theorem} \label{thm: 3 term no copy-cup}
    3-term group algebra codes do not have a valid pre-orientation that satisfies conditions for the 3-copy-cup gate as described in \cref{prop: lambda 3 conditions}.
\end{theorem}
\begin{proof}
    By \cref{eqn: parity CCZ non associative}, we can only assign ~$|\delta_{in}| = |\delta_{out}|=|\delta_{free}|=1$, i.e. the assignment~$(|\delta_{in}|,|\delta_{out}|,|\delta_{free}|)=(1,1,1)$. However, from \cref{eqn: associative cup product singular 1,eqn: associative cup product singular 2}, expansion of the~$|\delta_{out}(a_1)\cap \delta_{free}(a_2)|$ and $|\delta_{free}(a_1)\cap \delta_{in}(a_2)|$ in the singular conditions have single terms that cannot be paired. Therefore by properties of group algebras there can always be $a_1,a_2$ chosen such that these equations cannot be satisfied.
\end{proof}
This principle can be straightforwardly extended. 

\begin{lemma} \label{lemma: expansion}
    Consider the system of equations obtained from setting inequalities on $\{a_i\}$ from the master equations defined in \cref{eqn: master,eqn: master2}. Each equation involves expanding intersections of the coboundary partitions. Then,

    \begin{enumerate}
        \item if, for any expanded equation, the number of non-zero intersection terms is odd, or
        \item if coboundary partitions factored out in singular conditions as defined in \cref{defn: singular conditions} have only single terms,
    \end{enumerate}
    the system of equations cannot be satisfied, and the assignment does not have pre-orientation for the copy-cup gate.
\end{lemma}
\begin{proof}
    If expanded intersections have an odd number of terms, then we cannot pair terms up such that if $\{a_i\}$ are chosen where one of the terms is non-zero, the other will be non-zero as well. Therefore the system of equations cannot be satisfied, and we obtain the first statement.
    
    For the second statement, consider expanding intersection terms for equations involving singular conditions. One of the involved coboundary partitions only has a single term, say $g_1$. Then every expanded term will be of the form~$|g_1\cap \cdots|$ if the single-term coboundary partition is on the left, or~$|\cdots\cap g_1|$ if the partition is on the right. In the first case, pairing terms up will result in conditions of the form
    \begin{subequations}
        \begin{align}
            g_1a_1 &= g_ia_2,\\
            g_1a_1 &= g_ja_2\\
            \implies g_i&=g_j
        \end{align}
        which is a contradiction on the group algebra elements being distinct. The same analysis applies to the second case. Therefore the system of equations cannot be satisfied, and the assignment does not have pre-orientation for the copy-cup gate.
    \end{subequations}
\end{proof}
\begin{exa}{Contradiction}
    Consider the condition of \cref{eqn: 22+23=0}. If $\delta_{out} := g_1$ and $\delta_{free} := g_2+g_3$, expanding the intersection and keeping only non-zero terms returns the equation
    \begin{equation}
        |g_1a_1\cap g_2a_2|+|g_1a_1\cap g_3a_2)| = 0 \mod 2.
    \end{equation}
    There are two terms in the equation, and we pair them so that if $a_1,a_2$ are chosen such that $|g_1a_1\cap g_2a_2|=1$, we will have that $|g_1a_1\cap g_3a_2|=1$. However, pairing them leads to the condition
    \begin{equation}
        g_2=g_3,
    \end{equation}
    which is a contradiction on group elements of the code being distinct.\\

    Similarly, consider the conditions of \cref{eqn: 31=0 symmetric}, and where $|\delta_{in}|=|\delta_{out}|=3$. There are 9 terms in the expansion of the intersections. After pairing, there will be a single unpaired term of the form $|g_ia_1\cap g_ja_2|$. Then choosing $a_1 = g_i^{-1}g_ja_2$ will lead to the integrated Leibniz rule not being satisfied.
\end{exa}
Using \cref{lemma: expansion} we can immediately extend the discussion of the 3-copy-cup gate to group algebra codes with weight 5.
\begin{corollary}
    Weight 5 group algebra codes do not have non-trivial pre-orientation that satisfies conditions for the symmetric or non-associative 3-copy-cup gates described in \cref{prop: lambda 3 conditions}.
\end{corollary}
\begin{proof}
    The pre-orientation must satisfy the conditions listed in \cref{eqn: non associative cup product} or \cref{eqn: symmetric non associative cup product}. To satisfy the parity condition of \cref{eqn: parity CCZ non associative} or \cref{symmetric CCZ parity}, we assign
    \begin{equation}
        (|\delta_{in}|,|\delta_{out}|,|\delta_{free}|)\in \{(1,1,3),(2,2,1),(1,3,1), (3,1,1)\}.
    \end{equation}
    All assignments have $|\delta_{in}|=1$ or $|\delta_{out}|=1$ or $|\delta_{free}|=1$. Therefore in the expansion of the intersection~$|\delta_{free}(a_1)\cap \delta_{in}(a_2)|$ for \cref{eqn:singular1} in the non-associative cup product conditions, or \cref{eqn: 31=0 symmetric} in the symmetric cup product conditions, the number of terms will either be odd, or lead to contradictions on the group algebra elements when paired according to \cref{lemma: expansion}.

\end{proof}
\begin{remark}
    While we have shown that weight 3 and weight 5 group algebra codes cannot satisfy pre-orientation conditions for the 3-copy-cup gate, higher odd weight codes might still have valid pre-orientation as assignments are not immediately disqualified by \cref{lemma: expansion}. For example, a 7-term group algebra code has assignment $(|\delta_{in}|,|\delta_{out}|,|\delta_{free}|)=(2,2,3)$ where no partition has single terms, and so singular conditions are not immediately unsatisfiable.
\end{remark}
\subsection{Examples of codes with constant-depth $\operatorname{CZ}$ gates}
We have shown that 3-term group algebra codes cannot satisfy pre-orientation for the 3-copy-cup gate, but can satisfy that for 2-copy-cup gates. We attempt to find CSS codes constructed from the balanced product of weight 3 abelian group algebra codes that satisfy the conditions for pre-orientation, and have constant-depth $\operatorname{CZ}$ via the 2-copy-cup gate. The search is conducted in abelian groups of order $\{36,45,54,72\}$ to compare with the bivariate bicycle codes of \cite[Table 1]{bravyi2023highthreshold}. The constructed quantum codes have stabiliser check weight of 6. To reduce the search space, we fix the first term of the group algebra code to be the identity element 1, but do not restrict ourselves to bivariate groups. We find that, in most cases, codes from univariate cyclic groups have equivalent parameters to those found from bivariate or multivariate abelian groups. Results are summarised in \cref{table: weight 3 CZ}. Code distances are exact and solved by formulating distance-finding as a problem in integer programming. The calculation for obtaining gates in the copy-cup circuit are detailed in \cref{appendix: 2 copy-cup gate}, while the check for non-trivial logic is given in \cref{alg: 2copy nontrivial}.
\begin{table}[h]
\renewcommand{\arraystretch}{1.4}
\centering
\begin{tabular}{|c|c|c|c|}
\hline
$G$ & $[[n,k,d]]$  & $p_1$ & $p_2$\\\hline
$C_{9}\times C_4$ & [[72,8,6]] & $1+x^4+x^8$ & $1+x^{2}+xy^2$ \\\hline
$C_{9}\times C_4$ & [[72,4,8]] & $1+x^4y^3+x^8y^2$ & $1+x^5y^2+x^7y$ \\\hline
$C_{5}\times C_3\times C_3$ & [[90,8,6]] & $1+x^2y+x^4y^2$ & $1+x^3z^2+x^4z$ \\\hline
$C_{9}\times C_5$ & [[90,4,10]] & $1+x^4y^2+x^8y^4$ & $1+xy^4+x^5y^2$ \\\hline
$C_{27}\times C_2$ & [[108,12,6]] & $1+x^{12}+x^{24}$ & $1+x^3y+x^6$ \\\hline
$C_{27}\times C_2$ & [[108,4,10]] & $1+x^{13}y+x^{26}$ & $1+x^{23}+x^{25}y$\\\hline
$C_{9}\times C_8$ & [[144,16,6]] & $1+x^4+x^8$ & $1+xy^4+x^2$\\\hline
$C_{9}\times C_8$ & [[144,8,8]] & $1+x^4y^6+x^8y^4$ & $1+x^5y^4+ x^7y^2$\\\hline
$C_{9}\times C_8$ & [[144,4,12]] & $1+ x^4y^6 + x^8y^4 $ & $1+ x^4y^7+x^8y^6 $\\\hline
\end{tabular}
\caption{CSS codes found through a numerical search equipped with constant-depth $\operatorname{CZ}$ gates from the 2-copy-cup gate. Codes have $\operatorname{X}$ and $\operatorname{Z}$ check weight of 6, and are constructed as balanced products of weight 3 abelian group algebra codes with pre-orientation for the 2-copy-cup gate. The 2-copy-cup gate is simulated numerically according to \cref{alg: 2copy nontrivial}, and non-trivial action on the codespace is found indicating a constant-depth $\operatorname{CZ}$ gate. Groups presented have up to 3 generators $x,y,z$, whose orders are defined left-to-right by the subscripts on the products of cyclic groups in the first column.}
\label{table: weight 3 CZ}
\end{table}
\begin{remark} \label{rem: fundamental}
Almost all of the groups presented in the table admit a univariate presentation. For example, we have that $C_9\times C_4\cong C_{36},C_{27}\times C_2 \cong C_{54},$ and $C_9\times C_8\cong C_{72}$ as the input groups have coprime order, and by the fundamental theorem of finite abelian groups \cite[Chapter 5.2 Thm.\ 3]{Dummit_Foote_2004}. In the table, we present the codes as returned by a numerical search. A different presentation is possible by identifying appropriate group isomorphisms.
\end{remark}
\section{Group Algebra Codes with Higher Weight} \label{sec: higher weight}
In this section, we reason the existence of pre-orientations for group algebra codes with higher weight. Analysis follows from that for the 3-term group algebra codes. 
Finding conditions on the code can be generalised to a combinatorial argument. With higher weights, the calculation is unwieldy, and so we turn to numerical simulation to return valid configurations and conditions on the group elements. To ease notation in the next few sections, let us define tuples and triples of numbers as intersections of group elements. Pairs of tuples and triples can be written $(ij,kl)$ or $(ijk,lmn)$ that then impose conditions on group algebra elements. 
\begin{exa}{Notation}
Tuples and triples of numbers refer to group elements in double and triple intersection terms. For example, the triple $145$ refers to the intersection 
\begin{equation}
    g_1a_1\cap g_4a_2\cap g_5a_3\notag
\end{equation}
and the tuple $23$ refers to the intersection  
\begin{equation}
    g_2a_1\cap g_3a_2.,\notag
\end{equation}
where again we remind the reader that $g_i,a_i$ are $\mathbb{F}_2$-vectors that we freely identify with their support. Just as in notation for the coboundary partitions, the group elements act on basis elements $\{a_1,a_2,a_3\}\in X_0$ implicitly based on the position of the index.\\

Pairs of tuples and triples impose conditions on group algebra elements. For example, the tuple $(12,34)$ refers to $g_1a_1\cap g_2a_2$ and $g_3a_1\cap g_4a_2$ being satisfied at the same time for any $a_1,a_2$, which returns the condition 
\begin{equation}
    g_1^{-1}g_2=g_3^{-1}g_4. \notag
\end{equation}
The tuple $(123,456)$ refers to $g_1a_1\cap g_2a_2\cap g_3a_3$ being paired with $g_4a_1\cap g_5a_2\cap g_6a_3$, returning 2 independent conditions 
\begin{subequations}
    \begin{align}
        g_1^{-1}g_2&=g_4^{-1}g_5,\notag \\
        g_2^{-1}g_3&=g_5^{-1}g_6.
    \end{align}
\end{subequations}
\end{exa}
Following the previous sections, determining conditions on group algebra codes can be summarised as
\begin{enumerate} \label{pseudocode}
    \item Determine an assignment of terms in the group algebra coboundary operator to $\delta_{in},\delta_{out},\delta_{free}$.
    \item Expand out intersection terms for the system of equations for the copy-cup gate pre-orientation. For the 2-copy-cup gate, the system of equations is given by \cref{eqn: CZ preorientation conditions}, and is the same for both the symmetric and non-associative cup products. For the 3-copy-cup gates, the system of equations is given by \cref{eqn: non associative cup product} for the non-associative cup product, or \cref{eqn: symmetric non associative cup product} for the symmetric cup product. 
    \item For each equation in the system of equations, remove any 0 terms that arise from intersections with repeated group elements. For example, terms of the form $|g_ia_1\cap g_i a_2|$ are zero since $a_1\neq a_2$. Similarly, triple intersections are zero if they contain duplicated group elements, of the form $(i,i,j),(i,j,i),(j,i,i)$.
    \item If there are an odd number of remaining non-zero terms, or if there are singular conditions involving coboundary partitions with single terms, the condition cannot be satisfied by \cref{lemma: expansion}, and the assignment is invalid.
    \item Otherwise, pair all terms. A \emph{configuration} is defined to be a pairing of all terms for each equation. Terms are paired so if $a_1,a_2\in X_0$ are chosen such that a double intersection is non-zero, then the same $a_1,a_2$ will cause its paired double intersection to also be non-zero. Similarly, if $a_1,a_2,a_3\in X_0$ are chosen such that a triple intersection is non-zero, then the same $a_1,a_2,a_3$ will cause its paired triple intersection to also be non-zero.
    \item Determine if the configuration leads to contradictions in the group elements. If yes, then the configuration is not valid.
    \item Each pair of double intersections imposes a constraint on the group algebra elements. A tuple pair 
    \begin{equation}
        (12,34) \notag
    \end{equation} 
    indexing the pair of double intersections 
    \begin{equation}
        (g_1a_1\cap g_2a_2,g_3a_1\cap g_4a_2) \notag
    \end{equation}
    cannot have duplicated group elements $g_1=g_3$ or $g_2=g_4$ as they imply $g_1^{-1}g_2=g_1^{-1}g_4$ and $g_1^{-1}g_2=g_3^{-1}g_2$ respectively if both of them are to be simultaneously satisfied. 

    \item Each pair of triple intersections impose 2 constraints on the group algebra elements. A triple pair
    \begin{equation}
        (123,456), \notag
    \end{equation}
    indexing a pair of triple intersections
    \begin{equation}
        (g_1a_1\cap g_2a_2 \cap g_3a_3,g_4a_1\cap g_5a_2 \cap g_6a_3) \notag
    \end{equation}
    cannot have duplicated group elements $g_1=g_4,g_2=g_5,g_3=g_6$. Moreover, each new condition imposed on the group elements must not contradict conditions from other pairings. For example, while the pairs $(123,456)$ and $(125,461)$ are not contradictory on their own, they impose the restrictions $g_1^{-1}g_2 = g_4^{-1}g_5$ and $g_1^{-1}g_2=g_4^{-1}g_6$ respectively, which leads to the contradiction $g_5=g_6$.

    \item If there are no valid configurations, then the assignment is invalid. Else, there exist group algebra codes with pre-orientation based on the conditions of the valid configurations.
\end{enumerate}
This summarises and extends the analysis from the previous section for 3-term group algebra codes, and makes it amenable towards numerical implementation. The number of possible configurations scales as $n!/(n/2)!(2^{n/2})$ where $n$ is the number of terms to pair, and so can increase rapidly even for small group algebra codes. This problem is known as perfect matching in graph theory \cite{perfect_matching}. Pairs of terms can be interpret as graph edges, while the terms themselves are vertices. Constraints are updated dynamically, where selecting edges might render others invalid. Finding \emph{a} solution to the perfect matching problem can be done in polynomial time, such as by using Edmonds' blossom algorithm \cite{Edmonds_1965}, which has been used in other areas of quantum error correction, for example in minimum weight perfect matching decoders \cite{higgott2021pymatchingpythonpackagedecoding}. However, enumerating \emph{all} solutions to a perfect matching problem is computationally challenging.
\subsection{4-term group algebra codes}
For 4-term group algebra codes, the number of terms is still small enough to perform the calculation for pre-orientation by hand. In this section, we explicitly determine the conditions for pre-orientation for the 2- and 3-copy-cup gates on weight 4 group algebra codes. We do this analytically, and check the calculation using an implementation of the algorithm described in the previous section.
\subsubsection{Conditions for 2-copy-cup gates}
\begin{theorem} \label{thm: 2 copy cup}
    Consider a 4-term group algebra code with $\delta=g_1+g_2+g_3+g_4\in \mathbb{F}_2[G]$, following the definitions in \cref{sec: group algebra}. The code has pre-orientation satisfying the 2-copy-cup gate conditions of \cref{prop ryan ver} dependent on group elements assigned to partitions $\delta_{in},\delta_{out}$, and $\delta_{free}$. The pre-orientation must satisfy the conditions
    \begin{subequations} \label{eqn: CZ 112 assignment}
    \begin{align}
        g_3^{-1}g_1 &= g_2^{-1}g_3\text{ and }g_4^{-1}g_1=g_2^{-1}g_4, \text{or }\\
        g_3^{-1}g_1 &= g_2^{-1}g_4 \text{ and }g_4^{-1}g_1= g_2^{-1}g_3 , \label{eqn: eberhardt ex}\\
        \delta_{in}&:=g_1,\notag\\
        \delta_{out}&:=g_2,\notag\\
        \delta_{free}&:=g_3+g_4\notag
        \end{align}
    \end{subequations}
    or 
    \begin{subequations} \label{eqn: 11+22=0 assignment 220}
    \begin{align}
            g_1^{-1}g_2 &= g_3^{-1}g_4, \text{ or} \\
            g_1^{-1}g_2 &= g_4^{-1}g_3,\text{ or} \\
            g_1^{-1}g_2 &= g_2^{-1}g_1 \text{ and } g_3^{-1}g_4=g_4^{-1}g_3,\\
        \delta_{in}&:=g_1+g_2,\notag\\
        \delta_{out}&:=g_3+g_4,\notag\\
        \delta_{free}&:=\emptyset\notag
        \end{align}
    \end{subequations}
    or
    \begin{subequations} \label{eqn: 11+22=0 assignment 130}
    \begin{align}
            g_2^{-1}g_3 &= g_3^{-1}g_2, g_2^{-1}g_4=g_4^{-1}g_2 , g_4^{-1}g_3=g_3^{-1}g_4, \text{ or} \\
            g_2^{-1}g_3 &= g_3^{-1}g_2, g_2^{-1}g_4=g_4^{-1}g_3 , g_3^{-1}g_4=g_4^{-1}g_2, \text{ or}\\
            g_2^{-1}g_3 &= g_3^{-1}g_4, g_2^{-1}g_4=g_4^{-1}g_2 , g_3^{-1}g_2=g_4^{-1}g_3, \text{ or}\\
            g_2^{-1}g_3 &= g_4^{-1}g_2, g_2^{-1}g_4=g_3^{-1}g_2 , g_4^{-1}g_3=g_3^{-1}g_4,\\
        \delta_{in}&:=g_1,\notag\\
        \delta_{out}&:=g_2+g_3+g_4,\notag\\
        \delta_{free}&:=\emptyset\notag
        \end{align}
    \end{subequations}
    or
    \begin{subequations} \label{eqn: 11+22=0 assignment 310}
    \begin{align}
            g_1^{-1}g_2 &= g_2^{-1}g_1, g_1^{-1}g_3=g_3^{-1}g_1 , g_2^{-1}g_3=g_3^{-1}g_2, \text{ or} \\
            g_1^{-1}g_2 &= g_2^{-1}g_1, g_1^{-1}g_3=g_3^{-1}g_2 , g_2^{-1}g_3=g_3^{-1}g_1, \text{ or}\\
            g_1^{-1}g_2 &= g_2^{-1}g_3, g_1^{-1}g_3=g_3^{-1}g_1 , g_2^{-1}g_1=g_3^{-1}g_2, \text{ or}\\
            g_1^{-1}g_2 &= g_3^{-1}g_1, g_1^{-1}g_3=g_2^{-1}g_1 , g_2^{-1}g_3=g_3^{-1}g_2,\\
        \delta_{in}&:=g_1+g_2+g_3,\notag\\
        \delta_{out}&:=g_4,\notag\\
        \delta_{free}&:=\emptyset\notag
        \end{align}
    \end{subequations}
\end{theorem}
\begin{proof}
    We start with the general condition on the pre-orientation derived in \cref{eqn: CZ preorientation conditions}. To satisfy the parity condition, we must assign an even number of terms to the in/out terms, i.e. $|\delta_{in}|+|\delta_{out}|=0 \mod 2$. This gives us 3 assignments of coboundary partitions,
    \begin{equation}
        (|\delta_{in},|\delta_{out}|,|\delta_{free}|)\in \{(1,1,2),(2,2,0),(1,3,0),(3,1,0)\}.
    \end{equation}

    Let us consider the assignment $(1,1,2)$. By \cref{eqn: CZ preorientation conditions}, since we satisfy the no overlapping bits condition of \cref{nonoverlap}, we default to the case
    \begin{equation}
        |\delta_{free}(a_1)\cap \delta_{in}(a_2)|+|\delta_{out}(a_1)\cap \delta_{free}(a_2)|= 0\mod2.
    \end{equation}
    Expanding the terms out, we have that
    \begin{equation}
            |g_3a_1\cap g_1a_2|+|g_4a_1\cap g_1a_2| + |g_2a_1\cap g_3a_2| + |g_2a_1\cap g_4a_2 |= 0\mod2.
    \end{equation}
    Using tuples $ij$ to represent intersection terms $g_ia_1\cap g_ja_2$, we have the set of terms
    \begin{equation}
        31,41,23,24
    \end{equation}
    that need to be paired. There are 4 terms in this equation. Much like the proof of \cref{thm:2lambda}, we pair the terms up and determine conditions for the pairs to be satisfied at the same time. This is to ensure that, for any $a_1,a_2$ chosen, pairs of terms will evaluate to 1 at the same time. There are 3 distinct configurations of pairings. They are
    \begin{subequations}
        \begin{align}
            &\{(31,41),(23,24)\},\\
            &\{(31,23),(41,24)\},\\
            &\{(31,24),(41,23)\},
        \end{align}
    \end{subequations}
    which return conditions
    \begin{subequations}
        \begin{align}
            g_3^{-1}g_1 = g_4^{-1}g_1 &\text{ and } g_2^{-1}g_3=g_2^{-1}g_4, \text{or} \label{eqn:abc 1}\\
            g_3^{-1}g_1 = g_2^{-1}g_3 &\text{ and } g_4^{-1}g_1=g_2^{-1}g_4, \text{or} \label{eqn:abc 2}\\
            g_3^{-1}g_1 = g_2^{-1}g_4 &\text{ and } g_4^{-1}g_1=g_2^{-1}g_3 \label{eqn:abc 3}
        \end{align}
    \end{subequations}
    respectively. \cref{eqn:abc 1} is not satisfiable since we have assumed the group elements to be distinct; the other 2 equations return the conditions in \cref{eqn: CZ 112 assignment}.

    Now let us consider the assignment $(2,2,0)$ The free term is 0, and we are left with the condition
    \begin{equation} \label{eqn:blablalb}
        |\delta_{in}(a_1)\cap \delta_{in}(a_2)| + |\delta_{out}(a_1)\cap \delta_{out}(a_2)| = 0 \mod 2.
    \end{equation}
    We go through a similar analysis as before. The expansion of \cref{eqn:blablalb} results in 8 terms. However, 4 of the terms of the form $|g_ia_1\cap g_ia_2|$ are immediately 0 due to the rearrangement theorem. We are again left with 4 terms; written as tuples they are 
    \begin{equation}
            12,21,34,43.
    \end{equation}
    There are 3 configurations of paired terms,
    \begin{subequations}
        \begin{align}
            &\{(12,43),(21,34)\},\\
            &\{(12,34),(21,43)\},\\
            &\{(12,21),(34,43)\},
        \end{align}
    \end{subequations}
    returning conditions
    \begin{subequations}
        \begin{align}
            g_1^{-1}g_2 = g_4^{-1}g_3 &\text{ and } g_2^{-1}g_1=g_3^{-1}g_4,\label{bac2} \text{or} \\
            g_1^{-1}g_2 = g_3^{-1}g_4 &\text{ and } g_2^{-1}g_1=g_4^{-1}g_3, \text{or} \label{bac}\\
            g_1^{-1}g_2 = g_2^{-1}g_1 &\text{ and } g_3^{-1}g_4=g_4^{-1}g_3.
        \end{align}
    \end{subequations}
    Note that the conditions in \cref{bac,bac2} are repeated. Removing the repeats returns the conditions of \cref{eqn: 11+22=0 assignment 220}.
    
    Now consider the assignment $(1,3,0)$. Again we must satisfy \cref{eqn:blablalb}. Since $|\delta_{in}|=1$, the equation reduces to
    \begin{equation}
        |(g_2+g_3+g_4)a_1\cap(g_2+g_3+g_4)a_2 |= 0 \mod 2.
    \end{equation}
    Following the previous sections, we expand out the terms in the equation. There are 9 terms, of which 3 of the form $|g_ia_1\cap g_ia_2|$ are 0 due to the rearrangement theorem. The 6 remaining terms are
    \begin{equation}
        23,24,32,34,42,43.
    \end{equation}
    Now we are interested in configurations of pairings. Every configuration will return 3 conditions on the group algebra elements. However, some of the pairings will lead to contradictions. For example, the configuration
    \begin{equation}
        \{(23,32),(24,34),(43,42)\}
    \end{equation}
    is invalid due to the $(24,34)$ and $(43,42)$ pairs due to repeated numbers, since for example
    \begin{subequations}
        \begin{align}
            &(24,34)\\
            &\implies |g_2a_1\cap g_4a_2|=1, |g_3a_1\cap g_4a_2|=1\\
            &\implies g_2a_1=g_4a_2,g_3a_1=g_4a_2\\
            &\implies g_2^{-1}g_4=g_3^{-1}g_4\\
            &\implies g_2=g_3.
        \end{align}
    \end{subequations}
    We end up with 4 valid configurations that return conditions
    \begin{subequations}
        \begin{align}
            g_2^{-1}g_3 &= g_3^{-1}g_2, g_2^{-1}g_4=g_4^{-1}g_2 , g_4^{-1}g_3=g_3^{-1}g_4, \text{ or} \\
            g_2^{-1}g_3 &= g_3^{-1}g_2, g_2^{-1}g_4=g_4^{-1}g_3 , g_3^{-1}g_4=g_4^{-1}g_2, \text{ or}\\
            g_2^{-1}g_3 &= g_3^{-1}g_4, g_2^{-1}g_4=g_4^{-1}g_2 , g_3^{-1}g_2=g_4^{-1}g_3, \text{ or}\\
            g_2^{-1}g_3 &= g_4^{-1}g_2, g_2^{-1}g_4=g_3^{-1}g_2 , g_4^{-1}g_3=g_3^{-1}g_4,\\
        \end{align}
    \end{subequations}
    which is \cref{eqn: 11+22=0 assignment 130}.

    Finally, we note that the $(3,1,0)$ assignment is symmetric to the $(1,3,0)$ assignment due to symmetry in the 2-copy-cup gate pre-orientation conditions. We omit the analysis for this assignment. All of the conditions are further verified numerically via a perfect matching algorithm.
\end{proof}
\begin{remark}
    Some group algebra codes can satisfy multiple assignments. It is beneficial to use assignments that minimise $|\delta_{in}|$ and $|\delta_{out}|$ as the $\delta_{free}$ terms do not participate in the circuit for the copy-cup gate. Therefore reducing $|\delta_{in}|$ and $|\delta_{out}|$ results in a smaller number of physical gates.
\end{remark}
\begin{exa}{}
    The assignment $(1,1,2)$ and the conditions \cref{eqn: CZ 112 assignment} is a generalisation of the example given in \cite[Proposition 5.4]{breuckmann2024cupsgatesicohomology} to non-abelian weight 4 group algebra codes. Explicitly, the example was given as 
    \begin{subequations}
        \begin{align*}
            |\delta_{in}|& = 1,\\
            \delta_{out} &=\delta_{in}^{-1},\\
            \delta_{free} &= \delta_{free}^{-1},
        \end{align*}
    \end{subequations}
    where the superscript $^{-1}$ refers to the antipode map $a\mapsto a^{-1}$ being applied element-wise to group elements in the partitions. Setting 
        \begin{align*}
        \delta_{in}&:=g_1\\
        \delta_{out}&:=g_2\\
        \delta_{free}&:=g_3+g_4,
        \end{align*}
        these conditions are equivalent to 
        \begin{align*}
            g_1&=g_2^{-1},\\
            g_3+g_4 &=g_3^{-1}+g_4^{-1}.
        \end{align*} 
        The latter equation implies 
        \begin{align*}
            g_3=g_3^{-1} &\text{ and } g_4=g_4^{-1} \text{ or}\\
            g_3=g_4^{-1} &\text{ and } g_4=g_3^{-1}.
        \end{align*}
        We therefore have that
        \begin{align*}
            &g_1g_2=g_3^2 =g_4^2=1\\
            &\implies g_3^{-1}g_1=g_3g_2^{-1}\text{ and } g_4^{-1}g_1=g_4g_2^{-1} \text{ or}\\
            &g_1g_2=g_3g_4=g_4g_3=1\\
            &\implies g_3^{-1}g_1=g_4g_2^{-1} \text{ and } g_4^{-1}g_1=g_3g_2^{-1}
        \end{align*}
        which is equivalent to \cref{eqn: CZ 112 assignment} when the group is abelian.
    \end{exa}

\subsubsection{Examples of codes with constant-depth $\operatorname{CZ}$ gates}
We attempt to find weight 4 abelian group algebra codes that satisfy pre-orientation for the 2-copy-cup gate, and that have constant-depth $\operatorname{CZ}$ gates from the 2-copy-cup gate. A search is conducted through abelian groups of order $\{8,12,20,36\}$, and due to computational constraints, a partial search is conducted through the $C_{72}$ group. Optimal codes are presented in \cref{table: weight 4 CZ}.
\begin{table}[h!]
\renewcommand{\arraystretch}{1.4}
\centering
\begin{tabular}{|c|c|c|c|c|c|}
\hline
$G$ & $[[n,k,d]]$  & $p_1$ & $p_2$ & Assignment 1 & Assignment 2 \\\hline
$C_{8}$ & [[16,6,4]] & $1+x+x^2+x^3$ & $1+x+x^3+x^6$ & [2,2,0] & [1,1,2] \\\hline 
$C_{3}\times C_2\times C_{2}$ & [[24,2,6]] & $1+x+xy+x^2$ & $1+x^2+x^2y+x^2z$ & [1,1,2] & [1,3,0] \\\hline
$C_{4}\times C_3$ & [[24,4,5]] & $1+x^2+x^3+x^3y^2$ & $1+xy+x^2y^2+x^3$ & [3,1,0] & [1,1,2] \\\hline
$C_{5}\times C_{4}$ & [[40,2,8]] & $1+x^4+x^4y^2+x^4y^3$ & $1+x^3+x^3y^2+x^3y^3$ & [1,3,0] & [1,3,0] \\\hline
$C_{5}\times C_{4}$ & [[40,4,7]] & $1+xy+x^3y^3+x^4$ & $1+y^2+y^3+x^4y$ & [1,1,2] & [3,1,0] \\\hline
$C_{5}\times C_{4}$ & [[40,6,6]] & $1+y^2+y^3+x^4y$ & $1+y^2+y^3+x^3y$ & [3,1,0] & [3,1,0] \\\hline
$C_{9}\times C_4$  & [[72,6,10]] & $1+y^2+y^3+x^8y$ & $1+x^2y^2+x^6y^3+x^8y$ & [3,1,0] & [2,2,0] \\\hline 
$C_{9}\times C_4$ & [[72,8,8]] & $1+x^5y^3+x^6y^3+x^8$ & $1+x^5+x^6y+x^8y^3$ & [2,2,0] & [1,1,2] \\\hline
$C_{9}\times C_4$ & [[72,14,6]] & $1+x^5y^2+x^6y^2+x^8$ & $1+y^2+y^3+x^6y$ & [2,2,0] & [3,1,0] \\\hline
$C_{9}\times C_8$ & [[144,4,14]] & $1+x^6 +x^7y^2+x^8y^6$ & $1+x^2y^7+x^6y+x^8$ & [1,1,2] & [1,1,2] \\\hline 
$C_{9}\times C_8$ & [[144,6,12]] & $1+x^6+x^7y^2+x^8y^6$ & $1+x^2y+x^6y^6+x^8y^7$ & [1,1,2] & [2,2,0] \\\hline 
$C_{9}\times C_8$ & [[144,8,10]] & $1+x^6+x^7y^2+x^8y^6$ & $1+x^8+x^8y^2+x^8y^6$ & [1,1,2] & [1,3,0] \\\hline 
$C_{9}\times C_8$ & [[144,12,8]] & $1+x^6+x^6y^4+x^6y^6$ & $1+y^7+x^6y^4+x^6y^5$ & [1,3,0] & [1,1,2] \\\hline
\end{tabular}
\caption{CSS codes found through a numerical search equipped with constant-depth $\operatorname{CZ}$ gates from the 2-copy-cup gate. The codes have $\operatorname{X}$ and $\operatorname{Z}$ check weight of 8, and are constructed as the balanced product of two weight 4 abelian group algebra codes from $\mathbb{F}_2[G]$ satisfying pre-orientation for the 2-copy-cup gate. Groups presented have up to 3 generators $x,y,z$, whose orders are defined left-to-right by the subscripts on the products of cyclic groups in the first column. The 2-copy-cup gate is numerically simulated for the codes which are found through a numerical search, and codes where the copy-cup gate have non-trivial action on the codespace are presented. The assignments of group elements to $\delta_{in},\delta_{out},$ and $\delta_{free}$ that results in non-trivial action is stated.}
\label{table: weight 4 CZ}
\end{table}
\subsubsection{Conditions for 3-copy-cup gates}

Let us now focus on the analysis for the 3-copy-cup gate. Conditions for the non-associative and symmetric cup products are different, and so we analyse them separately. Indeed, it was mentioned in \cite{jacob2025singleshotdecodingfaulttolerantgates} that the weight 4 group algebra codes with pre-orientation satisfying the non-associative 3-copy-cup gate had distance 2, whereas in \cite{menon2025magictricyclesefficientmagic}, codes with pre-orientation satisfying the symmetric 3-copy-cup gate with higher distances were found. In this section we fully determine the conditions on the group algebra codes for pre-orientation, and also provide reasoning for the low distance of the former codes.

\begin{theorem} [Weight 4 Non-associative cup product] \label{thm: weight 4 non associative}
    Consider a 4-term group algebra code with $\delta=g_1+g_2+g_3+g_4\in \mathbb{F}_2[G]$ following the definitions of \cref{sec: group algebra}. The code has pre-orientation satisfying the non-associative 3-copy-cup gate conditions of \cref{prop: lambda 3 conditions} that is dependent on group elements assigned to partitions $\delta_{in},\delta_{out}$, and $\delta_{free}$. The pre-orientation must satisfy the conditions
\begin{subequations} \label{eqn: 4 term CCZ}
    \begin{align}
            g_1^{-1}g_2&=g_2^{-1}g_1, \label{eqn: 4 term g1g2}\\
            g_3^{-1}g_4&=g_4^{-1}g_3, \\
            g_4^{-1}g_2 &= g_3^{-1}g_1, \label{eqn: 4 term g2g3}\\
        \delta_{in}&:=g_1+g_2,\notag\\
        \delta_{out}&:=g_3+g_4,\notag\\
        \delta_{free}&:=\emptyset.\notag
        \end{align}
    \end{subequations}
\end{theorem}
\begin{proof}
    By \cref{eqn: parity CCZ non associative}, valid coboundary assignments are
    \begin{equation}
        (|\delta_{in}|,|\delta_{out}|,|\delta_{free}|)\in \{(1,1,2),(2,2,0),(1,3,0),(3,1,0)\}.
    \end{equation}
    In the discussion of these assignments, we assign $g_i$ to the partitions starting from the lowest $i$ to $\delta_{in}$, then $\delta_{out}$, and finally $\delta_{free}$.

    Assignment $(1,1,2)$ is immediately invalid due to \cref{lemma: expansion} and the singular condition of \cref{eqn:singular1}. Explicitly, the assignment has to satisfy the condition
    \begin{equation}
            |g_1a_1\cap (g_3+g_4)a_2| = 0\mod2,
    \end{equation}
    Pairing these terms leads to a contradiction $g_3=g_4$.
    
    Consider the assignment $(3,1,0)$. Of the equations in \cref{eqn: non associative cup product}, the only relevant equations to satisfy are \cref{eqn:singular2,eqn: 111331 blabla bla}; that is, 
    \begin{subequations}
        \begin{align}
            |\delta_{in}(a_1)\cap \delta_{in}(a_2)|&=0\mod2,\\
            |\delta_{in}(a_1)\cap \delta_{in}(a_2)\cap\delta_{in}(a_3)| &=0\mod2.
        \end{align}
    \end{subequations}
    Upon expansion, the first equation has 9 terms while the second has 27. Both are reduced to 6 terms when removing terms that are zero, where there are repeated group elements in the intersection, which we remove due to the rearrangement theorem. The triple intersection has terms
    \begin{equation}
        123,132,213,231,312,321.
    \end{equation}
    We pair these 6 terms. However, repeated numbers on the same indices lead to contradictions on the group elements. Upon further inspection, there is no configuration where all pairs of triples have no repeated number at the same index. Therefore this assignment does not lead to pre-orientation.

   Now consider the assignment (1,3,0), which in this case is different from (3,1,0) unlike the case $\Lambda=2$ as there is no longer symmetry in the conditions. We focus on satisfying the equation
    \begin{equation}
        |\delta_{out}(a_1)\cap \delta_{out}(a_2)\cap\delta_{in}(a_3)| =0\mod2
    \end{equation}
    from \cref{eqn: 111331 blabla bla}. Expanding the equation results in 6 terms that must be paired, explicitly given as
    \begin{equation}
        231+241+321+341+421+431. 
    \end{equation}
    However, note that this is a singular condition since $|\delta_{in}|=1$, presenting itself on the repeated number on the third index on all terms. Pairing any of these terms lead to a contradiction, and therefore the assignment does not lead to a pre-orientation.

    Finally, consider the assignment (2,2,0). The equations to satisfy are \cref{eqn:singular2,eqn: 111331 blabla bla,eqn: 22+23=0}.
    \begin{subequations}
        \begin{align}
            |(g_1+g_2)a_1\cap (g_1+g_2)a_2| &=0\mod2\label{blabla1} \\\implies g_1^{-1}g_2=g_2^{-1}g_1\notag\\
            |(g_3+g_4)a_1\cap (g_3+g_4)a_2| &=0\mod2\label{blabla3} \\\implies g_3^{-1}g_4=g_4^{-1}g_3\notag\\
            |(g_3+g_4)a_1\cap (g_3+g_4)a_2\cap (g_1+g_2)a_3| &=0\mod2. \label{eqn: associative cup product CCZ difference}
        \end{align}
    \end{subequations}
    Expanding \cref{eqn: associative cup product CCZ difference} returns 8 terms. Of these, 4 of the form $g_ia_1\cap g_i a_2\cap g_ja_3$ are immediately 0, leaving us with 4 terms written as
    \begin{equation}
        341,342,431,432.
    \end{equation}
    The only valid configuration is $\{(341,432),(342,431)\}$, leading to the conditions
        \begin{subequations} \label{blabla2}
        \begin{align}
            g_3^{-1}g_4 &= g_4^{-1}g_3,\\
            g_4^{-1}g_1 &= g_3^{-1}g_2,\\
            g_4^{-1}g_2 &= g_3^{-1}g_1.
        \end{align}
    \end{subequations}
    Together, \cref{blabla1,blabla3,blabla2} lead to \cref{eqn: 4 term CCZ}, and hence the theorem.
\end{proof}
\begin{remark}
    The conditions of \cref{eqn: 4 term CCZ} have already been stated in \cite{jacob2025singleshotdecodingfaulttolerantgates}. Here we show that, in fact, this is the \emph{only} assignment of weight 4 group algebra codes -- abelian and non-abelian -- that can satisfy the pre-orientation for the non-associative triple cup product.
\end{remark}
It was mentioned in \cite{jacob2025singleshotdecodingfaulttolerantgates} that the 4-term group abelian group algebra codes they found were of distance 2. We show this is true for all codes that satisfy this pre-orientation.
\begin{theorem} \label{thm: 220 max dist 2}
    4-term abelian group algebra codes with $\delta = g_1+g_2+g_3+g_4\in \mathbb{F}_2[G]$ following the definitions of \cref{sec: group algebra} that satisfy the pre-orientation condition in \cref{eqn: 4 term CCZ} of \cref{thm: weight 4 non associative} will have an upper-bounded distance of 2. 
\end{theorem}
\begin{proof}
By \cref{eqn: 4 term CCZ}, we have that the group elements satisfy the equations
\begin{equation}
        g_1g_2^{-1}=g_2g_1^{-1}=g_3g_4^{-1}=g_4g_3^{-1}.
\end{equation}
Then we can rearrange the group algebra elements as
\begin{subequations}
    \begin{align}
        g_1+g_2+g_3+g_4&=g_1(1+g_1^{-1}g_2+g_1^{-1}g_3+g_1^{-1}g_4)\label{eqn: 4term eq1}\\
        &= g_1(1+g_1^{-1}g_2 + g_1^{-1}g_3+g_1^{-1}g_2g_1^{-1}g_3)\\
        &= g_1(1+g_1^{-1}g_2 +(1+g_1^{-1}g_2)g_1^{-1}g_3)\\
        &= g_1(1+g_1^{-1}g_2)(1+g_1^{-1}g_3) \label{eqn: abelian commute blabla}\\
        &= g_1(1+g_1^{-1}g_3)(1+g_1^{-1}g_2) \label{eqn: abelian commute blabla2},
    \end{align}
\end{subequations}
where we use the commutative property of abelian groups to commute the $(1+g_1^{-1}g_2)$ term across the $(1+g_1^{-1}g_3)$ term from \cref{eqn: abelian commute blabla} to \cref{eqn: abelian commute blabla2}. Then, since $g_1^{-1}g_2$ is an \emph{involution}, which are group elements with order 2, we have that
\begin{equation}
    \delta(1+g_1^{-1}g_2)=0.
\end{equation}
Therefore $(1+g_1^{-1}g_2)$ is a valid codeword, and so abelian group algebra codes that satisfy this pre-orientation will have an upper bound distance of 2.
\end{proof}
\begin{remark}
    The proof of \cref{thm: 220 max dist 2} relies on the fact that the codes are defined on abelian group algebras. This let us commute the $1+g_1^{-1}g_2$ term across the $1+g_1^{-1}g_3$ term to the right when going from \cref{eqn: abelian commute blabla} to \cref{eqn: abelian commute blabla2}. This naturally raises the question about how tensor product quantum codes constructed from non-abelian group algebra codes perform when subject to the same pre-orientation conditions.
\end{remark}
By \cref{thm: 220 max dist 2}, the classical abelian group algebra codes that satisfy the non-associative 3-copy-cup gate pre-orientation have maximum distance 2. It follows that product quantum codes obtained from weight 4 group algebra codes will have a max distance of 2 by the Kunneth formula, which relates the classical codewords to quantum codewords. This suggests that the non-associative cup product cannot be used with abelian group algebra codes to obtain high-distance quantum codes with constant-depth $\operatorname{CCZ}$ gates.

Remarkably, the different set of conditions using the symmetric cup product results in codes with distances greater than 2. The conditions for pre-orientation are summarised below.
\begin{theorem} [Weight 4 Symmetric cup product] \label{thm: weight 4 symmetric}
    Consider a 4-term group algebra code with $\delta=g_1+g_2+g_3+g_4\in\mathbb{F}_2[G]$ following the definitions of \cref{sec: group algebra}. The code has pre-orientation satisfying the symmetric 3-copy-cup gate conditions described in \cref{prop: lambda 3 conditions} dependent on group elements assigned to partitions $\delta_{in},\delta_{out}$, and $\delta_{free}$. The pre-orientation must satisfy the conditions
\begin{subequations} \label{eqn: CCZ 220 symmetric}
    \begin{align}
            g_1^{-1}g_2&=g_2^{-1}g_1 \text{ and } g_3^{-1}g_4=g_4^{-1}g_3,\text{ or } \label{eqn: menon conditions2}\\
            g_1^{-1}g_2&=g_3^{-1}g_4,\text{ or }\label{eqn: menon conditions}\\
            g_1^{-1}g_2&=g_4^{-1}g_3,\\
        \delta_{in}&:=g_1+g_2,\notag\\
        \delta_{out}&:=g_3+g_4,\notag\\
        \delta_{free}&:=\emptyset.\notag
        \end{align}
    \end{subequations}
\end{theorem}

\begin{proof}
    By \cref{symmetric CCZ parity}, we can only have coboundary assignments 
    \begin{equation}
        (|\delta_{in}|,|\delta_{out}|,|\delta_{free}|)\in \{(1,1,2),(2,2,0),(1,3,0),(3,1,0)\}.
    \end{equation}
    An analysis similar to that of \cref{thm: weight 4 non associative} shows that the assignments $(1,1,2),(3,1,0),$ and $(1,3,0)$ cannot be satisfied. In fact, in this case, conditions are symmetric, and so there is no need to analyse the $(1,3,0)$ and $(3,1,0)$ assignments separately.

    Let us focus on the $(2,2,0)$ assignment as stated in \cref{eqn: CCZ 220 symmetric} and how the conditions are obtained. The equations to satisfy are \cref{eqn: 11+22 symmetric,eqn: 111 blablabla symmetric}. Explicitly, \cref{eqn: 11+22 symmetric} is the equation
    \begin{equation}\label{eqn: symmetric blabla}
            |\delta_{in}(a_1)\cap \delta_{in}(a_2)| + |\delta_{out}(a_1)\cap \delta_{out}(a_2)| =0\mod2.
        \end{equation}
    Now consider \cref{eqn: 111 blablabla symmetric}. Upon inspection, we find that due to an empty $\delta_{free}$ partition and the fact that $|\delta_{in}|,|\delta_{out}|=2$, the equation is immediately satisfied, and there are no triple intersection terms to pair. Then we find that \cref{eqn: symmetric blabla} is the same as for the 2-copy-cup gate on weight 4 group algebra codes in \cref{eqn:blablalb}, returning the conditions of \cref{eqn: 11+22=0 assignment 220} and therefore proving the theorem. 
\end{proof}
\begin{remark}
    Conditions of \cref{eqn: CCZ 220 symmetric} generalise those stated in \cite[Theorem D.1]{menon2025magictricyclesefficientmagic}. We show this assignment is the only possible one for weight 4 group algebra codes, and extend discussion to non-abelian group algebra codes.
\end{remark}
\begin{remark}
    The symmetric cup product causes cancellation of triple intersection terms, leading to less restrictive conditions on pre-orientation compared to the non-associative cup product. It can also be shown that the outside-in method of computing the cup product will have 8 triple intersection terms to pair in \cref{fat eqn outside in} for the (2,2,0) assignment after expansion. which is extremely restrictive. The cancellation of triple intersection terms provides reasoning for the empirical observation made in \cite{menon2025magictricyclesefficientmagic} about codes that have pre-orientation for the symmetric 3-copy-cup gate having distances greater than 2.
\end{remark}
 Comparing \cref{eqn: CCZ 220 symmetric} and \cref{eqn: 4 term CCZ}, we see that pre-orientations are less restrictive. Specifically, we do not have the condition that the group algebra elements must make up involutions, and that the involutions \emph{must be equal to each other.} These conditions led to the distance bound in \cref{thm: 220 max dist 2} due to commutation of group elements and the ability to factor out $(1+g_1^{-1}g_2)$. The less restrictive conditions come from not needing to pair triple intersection terms due to the symmetry of this cup product as opposed to the non-associative case.
 
 In fact, we see that \cref{eqn: CCZ 220 symmetric} is equivalent to \cref{eqn: 11+22=0 assignment 220} for the 2-copy-cup gate. This is because required conditions for pre-orientation is the same as that for the 2-copy-cup gate, and that means that the codes will be amenable to both 2- and 3-copy-cup gates, and therefore have constant-depth $\operatorname{CZ}$ and $\operatorname{CCZ}$ gates.
 \begin{corollary} \label{cor: 3 and 2 have preorientation}
     4-term group algebra codes that satisfy pre-orientation conditions for the symmetric 3-copy-cup gate also satisfy conditions for the 2-copy-cup gate. 
 \end{corollary}
We can straightforwardly apply \cref{cor: 3 and 2 have preorientation} to the codes of \cite{menon2025magictricyclesefficientmagic}. We do so for the codes presented in \cref{table: menon codes CZ}, a subset of codes from \cite[Table 2]{menon2025magictricyclesefficientmagic} with $\leq 250$ qubits. We check that the 2-copy-cup gate preserves the codespace, and furthermore has non-trivial logical action according to \cref{alg: 2copy nontrivial}. %
\begin{table}[h]
\centering
\renewcommand{\arraystretch}{1.4}
\begin{tabular}{|c|c|c|c|c|}
\hline
$\,[[n, k, (d_X, d_Z)]]\,$ & $(l,m,n)$ & $a$ & $b$ & $c$ \\
\hline
$[[48, 6, (8,4)]]$ 
& $(2,2,4)$ 
& $y+z+xz+xyz^{2}$ 
& $yz^{2}+yz^{3}$ 
& $y+xyz$ \\
\hline
$[[84, 6, (12,5)]]$ 
& $(2,2,7)$ 
& $y+z+xz+xyz^{2}$ 
& $z^{3}+xz^{4}$ 
& $y+yz^{4}$\\
\hline
$[[108,6,(12,6)]]$ 
& $(3,3,4)$ 
& $x+z^{2}+yz+x^{2}yz^{3}$ 
& $y^{2}z+x^{2}yz^{3}$ 
& $x^{2}+x^{2}yz^{2}$\\
\hline
$[[240,6,(\le 22,8)]]$ 
& $(4,4,5)$ 
& $xy^{2}z^{3}+xy^{3}z^{4}+x^{2}y^{2}z+x^{2}y^{3}z^{2}$ 
& $y^{3}+x^{2}yz^{2}$ 
& $xz^{4}+x^{3}y^{3}z$\\
\hline
$[[108,12,(6,4)]]$ 
& $(3,3,4)$ 
& $z+xz^{3}+xyz^{2}+x^{2}y$ 
& $y^{2}+y^{2}z^{3}+xy^{2}z+x^{}y^{2}z^{2}$ 
& $z+xyz^{3}$ \\
\hline
$[[180,12,(15,6)]]$ 
& $(3,4,5)$ 
& $yz^{3}+y^{3}+x^{2}yz^{3}+x^{2}y^{3}z$ 
& $xy^{}z^{4}+xy^{2}z^{2}+x^{2}yz+x^{2}y^{2}z^{4}$ 
& $z^{4}+x^{2}z$ \\
\hline
$[[108,15,(12,6)]]$ 
& $(3,3,4)$ 
& $y+y^{2}z+xyz^{3}+x^{2}y^{2}z^2$ 
& $z^2+xy+xy^2z+x^2z^3$ 
& $yz^3+y^2z+x^2+x^2y^2z^2$ \\
\hline
\end{tabular}
    \caption{Codes from \cite[Table 2]{menon2025magictricyclesefficientmagic} with $n\leq250$. Polynomials $a,b,c\in \mathbb{F}_2[G]$ are the check polynomials of the classical group algebra code with $G=C_l\times C_m\times C_n$ generated by $x,y,z$ respectively. They are used to construct the 3-fold balanced product cube complex, which is collapsed into a total complex to define the trivariate tricycle code with qubits defined on the basis of 1-cochains.  
    The $\operatorname{X}$ and $\operatorname{Z}$ distances $d_X,d_Z$ are stated explicitly. Applying \cref{cor: 3 and 2 have preorientation}, the weight 4 group algebra codes that have pre-orientation for the symmetric 3-copy-cup gate also satisfy pre-orientation for the 2-copy-cup gate. The 2-copy-cup gate is numerically simulated, and shown to have non-trivial action indicating constant-depth $\operatorname{CZ}$ gates on all codes in the table.}
    \label{table: menon codes CZ}
\end{table}
\section{Further numerical simulation} \label{sec: further}
In this section, we present further results of numerical searches for new codes that satisfy pre-orientation for the 3-copy-cup gate. The result of \cref{thm: 220 max dist 2} motivates us to use the conditions for the symmetric cup product in our search for weight 4 group algebra codes. However, by results of \cref{thm: weight 4 symmetric}, it turns out that such a search has already been done in the work of \cite{menon2025magictricyclesefficientmagic} through general abelian group algebra codes. Moreover, in \cite{jacob2025singleshotdecodingfaulttolerantgates}, a search was also conducted for trivariate tricycle codes optimised on rate and distance, and so the numerical search for codes in these bodies of work have already been extensive. We attempt to differentiate our search from existing ones by being more targeted in the search parameters.
\subsection{Balanced products of univariate cyclic groups}
In both \cite{menon2025magictricyclesefficientmagic} and \cite{jacob2025singleshotdecodingfaulttolerantgates}, an emphasis is placed on so-called \emph{trivariate tricycle} codes, due to the fact that group algebra codes they consider come from abelian groups with 3 generators; or equivalently polynomial rings with 3 variables. However, in constructing the balanced product cochain complex, this condition is not required. Indeed, the isomorphism 
\begin{equation*}
    R\otimes_RR\cong R,
\end{equation*}
of \cref{eqn: RR isomorphism} has been pointed out for example in \cite[Equation 89]{menon2025magictricyclesefficientmagic} and \cite[Appendix A]{eberhardt2024logicaloperatorsfoldtransversalgates}, which can be used to simplify $\Lambda$-fold tensor product cochain complexes to have all elements defined over $R$ rather than $R\otimes_R\cdots \otimes_R R$. Moreover, finite abelian groups can be written as products of cyclic groups; this is known as the fundamental theorem of finite abelian groups \cite[Chapter 5.2 Thm.\ 3]{Dummit_Foote_2004}. Some of the examples found in \cite[Table~2]{menon2025magictricyclesefficientmagic} have a simpler group presentation. The $[[180,12,(15,6)]]$ code, for example, is constructed with group algebra elements from $\mathbb{F}_2 [C_3\times C_4\times C_5]$. However, all of the univariate group orders are coprime. Therefore this group is isomorphic to the univariate cyclic group algebra $\mathbb{F}_2[C_{60}]$. Many of the trivariate tricycle codes can therefore be defined equivalently on bivariate, or in this case univariate cyclic groups. This identification might make codes more amenable to analysis, since properties of rate and distance for univariate cyclic group algebra codes are well-understood. In turn this might lead to a smaller dependence on brute force numerical searches.
\begin{exa}{Collapsing Toric Codes}
    Consider square and cube complexes from the 2-fold and 3-fold tensor products of 2-term cochain complexes of the form
    \begin{center}
    \begin{tikzcd}[cells={nodes={minimum height=2em}}]
        R \arrow[r, "\delta\in R"] & R, R:=\mathbb{F}_2[G].
    \end{tikzcd}
    \end{center}
    Quantum codes defined on square complexes can be equipped with 2-copy-cup gates, while those on cube complexes can be equipped with 3-copy-cup gates. Bivariate bicycle codes \cite{bravyi2023highthreshold} are an example of abelian two-block group algebra codes \cite{pryadko_2BGA} on the square complex, taking the tensor product $\otimes_R$ of two 2-term chain complexes with $R=\mathbb{F}_2[C_l\times C_m]$. Meanwhile trivariate tricycle codes are defined on a cube complex by taking the tensor product $\otimes_R$ of three 2-term chain complexes with $R=\mathbb{F}_2 [C_l\times C_m\times C_n]$.\\

    Now consider the case where $G$ is the univariate cyclic group. When the tensor product is taken over $\mathbb{F}_2$, i.e. in the hypergraph product construction, then the 2-fold and 3-fold tensor products of the cochain complexes will necessarily have sectors on the hypercube complex be generated by 2 and 3 variables respectively. These sectors are placed on vertices of the hypercube complex, and will be of the form $\mathbb{F}_2 [C_l\times C_m]$ or $\mathbb{F}_2 [C_l\times C_m\times C_n]$. For example, taking the 3-fold tensor product $\otimes_{\mathbb{F}_2}$ of repetition codes from the cyclic group $C_l$ returns the 3-dimensional toric code of length $l$. Sectors positioned on vertices of the cube complexes as shown in \cref{exa: cube complexes} will be of the form
    \begin{subequations}
    \begin{align}
     (\mathbb{F}_2 [C_l])\otimes_{\mathbb{F}_2}(\mathbb{F}_2 [C_l])\otimes_{\mathbb{F}_2}(\mathbb{F}_2 [C_l]) \cong \mathbb{F}_2[x,y,z]/\langle x^l+1,y^l+1,z^l+1\rangle. \notag
    \end{align}
    \end{subequations}
    On the other hand, we could also take the tensor product $\otimes_{R}$ where $R:=\mathbb{F}_2[C_l]$. This can be interpreted geometrically as ``collapsing" the 3-dimensional torus into a single dimension. Then on the cube complex we have sectors
    \begin{subequations}
    \begin{align}
        (\mathbb{F}_2 [C_l])\otimes_{\mathbb{F}_2 [C_l]}(\mathbb{F}_2 [C_l])\otimes_{\mathbb{F}_2 [C_l]}(\mathbb{F}_2 [C_l]) &\cong \mathbb{F}_2[C_l]\notag\\
        &\cong \mathbb{F}_2[x]/\langle x^l+1\rangle. \notag
    \end{align}
    \end{subequations}
    The sectors are generated by single variables, and are in fact cyclic group algebras. Generalising, we see that under the balanced product, there is no need for the trivariate restriction in choosing $G$, which can be any abelian group, in construction of hypercube complexes.
\end{exa}
Let us build on this observation and perform a numerical search for group algebra codes in \emph{univariate} cyclic groups, in the construction of CSS codes amenable to the symmetric 3-copy-cup gate. In particular, we focus on the repetition code defined by check polynomial $1+x$, where $x$ generates the group. We attempt to find \emph{collapsed} 3D toric codes with optimal parameters, and that satisfy the conditions of \cref{eqn: CCZ 220 symmetric} for the symmetric 3-copy-cup gate. Results are presented in \cref{table: search results}.

\begin{table}[h]
\renewcommand{\arraystretch}{1.4}
\centering
\begin{tabular}{|c|c|c|c|c|}
\hline
$G$ & $[[n,k,d]]$  & $p_1$ & $p_2$ & $p_3$ \\\hline
$C_{4}$ & [[12,3,2]] & $1+x+x^2+x^{4}$ & $1+x+x^2+x^{4}$ & $1+x+x^2+x^{4}$\\\hline
$C_{7}$ & [[21,3,3]]  & $1+x+x^2+x^{3}$ & $1+x^{}+x^{3}+x^{4}$ & $1+x^{2}+x^{3}+x^{5}$\\\hline
$C_{11}$ & [[33,3,4]] & $1+x+x^2+x^{3}$ & $1+x^{}+x^{3}+x^{4}$ & $1+x^{2}+x^{4}+x^{6}$\\\hline
$C_{13}$ & [[39,3,5]]& $1+x+x^2+x^{3}$ & $1+x^{}+x^{3}+x^{4}$ & $1+x^{2}+x^{5}+x^{7}$\\\hline
$C_{21}$ & [[63,3,6]] & $1+x^{}+x^{2}+x^{3}$ & $1+x^{}+x^{5}+x^{6}$ & $1+x^{2}+x^{8}+x^{10}$\\\hline
\end{tabular}
\caption{Parameters of collapsed 3D toric codes found through numerical search. These are CSS codes with $\operatorname{X}$ check weight 12, and $\operatorname{Z}$ check weight 8. The codes are constructed as 3-fold balanced products of abelian group algebra codes from univariate cyclic group algebra $\mathbb{F}_2[G]$; the search was conducted for $|G|$ up to order 21. The polynomials come from the ideal $\langle 1+x\rangle$ in the ring $\mathbb{F}_2[G]/\langle x^n+1\rangle$, and therefore are different representative polynomials of the repetition code. Code polynomials satisfy conditions for pre-orientation; however, the 3-copy-cup gate results in trivial action on the codespace for all codes in this table. }
\label{table: search results}
\end{table} 
While the classical codes satisfy the pre-orientation conditions, we find that, for these codes, applying the copy-cup gate circuit results in trivial action on the logical space. These codes therefore represent optimal code parameters that have pre-orientation, but trivial logic from the 3-copy-cup gate.
\subsection{Searching for abelian group algebra codes}
We conduct a numerical search for classical group algebra codes with the following conditions:
\begin{itemize}
    \item They are weight 4 and satisfy pre-orientation conditions for the symmetric 3-copy-cup gate. The first element is fixed to be identity.
    \item Group orders are between 2 and 16. %
    \item The codes have dimension $k\leq 3$. This is motivated by a tradeoff between code rate and distance.
\end{itemize}
Quantum codes are constructed as the 3-fold balanced product of the group algebra codes if there are $\geq 3$ codes from the search. The code parameters are determined, and then a check is done for non-trivial logic from the 3-copy-cup gate. The code is recorded in \cref{table: non trivial logic} if this is satisfied. 
\begin{table}[h]
\renewcommand{\arraystretch}{1.4}
\centering
\begin{tabular}{|c|c|c|c|c|}
\hline
$G$ & $[[n,k,d]]$  & $p_1$ & $p_2$&$p_3$\\\hline
$C_9$ & $[[27,9,2]]$ & $1+x+x^3+x^4$ & $1+x+x^6+x^7$ & $1+x^2+x^3+x^5$\\\hline
\end{tabular}
\caption{Quantum codes with non-trivial 3-copy-cup gates found from a numerical search. The search is conducted for classical group algebra codes with pre-orientation for the 3-copy-cup gate, in all abelian groups of order between 2 and 16, and with code dimension $k\leq 3$. A single code was found to satisfy these constraints.}
\label{table: non trivial logic}
\end{table}
From the search, only a single code was found. In all other cases, either there were no abelian group algebra codes that satisfied the search parameters, or the constructed quantum codes had trivial logic from the 3-copy-cup gate.

We make an effort to reconcile this result with the $[[48,6,4]]$ code found in \cite{menon2025magictricyclesefficientmagic}, which has also been restated in \cref{table: menon codes CZ}. The code is constructed from a balanced product of group algebra codes, defined by a weight 4 polynomial and two weight 2 polynomials from the $C_2\times C_2 \times C_4$ abelian group. As classical group algebra codes, these are check polynomials for a $[16,8,4]$ and two $[16,4,4]$ cyclic codes respectively. This suggests that the parameter sweep conducted in our search is not thorough enough to find all interesting examples of balanced product group algebra codes equipped with a non-trivial copy-cup gate.
\subsection{2-term check polynomials from cyclic groups}
We can compare the quantum codes constructed from weight 4 codes obtained in \cref{table: search results} to codes constructed from polynomials of weight 2. Weight 2 check polynomials can be understood as variants of the repetition code. These have particularly simple form: for a cyclic group of order $n$ generated by $x$, the check polynomial $1+x^i$ generates the same code as $1+x$ if $i$ is coprime to $n$. This significantly reduces the search space. We construct 3D balanced product codes based on these polynomials and report them in \cref{table: search results2}.

\begin{table}[h]
\renewcommand{\arraystretch}{1.4}
\centering
\begin{tabular}{|c|c|c|c|c|}
\hline
$G$ & $[[n,k,d]]$  & $p_1$ & $p_2$ & $p_3$ \\\hline
$C_{2}$ & [[6,3,2]] & $1+x$ & $1+x$ & $1+x$\\\hline
$C_{7}$ & [[21,3,3]] & $1+x$ & $1+x^2$ & $1+x^3$\\\hline
$C_{14}$ & [[42,3,4]] & $1+x$ & $1+x^3$ & $1+x^5$\\\hline
$C_{27}$ & [[81,3,5]] & $1+x$ & $1+x^4$ & $1+x^{10}$\\\hline
\end{tabular}
\caption{Parameters of collapsed 3D toric codes found from a numerical search. The codes are constructed as the 3-fold balanced product of weight 2 repetition code polynomials. Groups up to $C_{30}$ have been tested. Codes with highest distances first appearing are presented. These codes are equipped with constant-depth $\operatorname{CCZ}$ from the 3-copy-cup gate, which has been implemented numerically.}
\label{table: search results2}
\end{table} 
We can further compare these to the trivariate tricycle codes from \cite[Table III]{jacob2025singleshotdecodingfaulttolerantgates}, also constructed from tensor products of weight 2 polynomials, presented in \cref{table: search results3}. By ignoring the trivariate restriction, we are able to achieve better code parameters at smaller code lengths.
\begin{table}[h]
\renewcommand{\arraystretch}{1.4}
\centering
\begin{tabular}{|c|c|c|c|c|}
\hline
$G$ & $[[n,k,d]]$  & $p_1$ & $p_2$ & $p_3$ \\\hline
$C_{3}\times C_2\times C_2$ & [[36,3,3]] & $1+xyz$ & $1+x^2z$ & $1+x^2y$\\\hline
$C_{4}\times C_2\times C_2$ & [[48,3,4]] & $1+x$ & $1+xz$ & $1+xy$\\\hline
$C_{5}\times C_3\times C_2$ & [[90,3,5]] & $1+x$ & $1+xy$ & $1+x^2y^2z$\\\hline
\end{tabular}
\caption{Optimal trivariate tricycle codes from \cite[Table III]{jacob2025singleshotdecodingfaulttolerantgates}, constructed as the 3-fold balanced product of weight 2 abelian group algebra codes.}
\label{table: search results3}
\end{table} 

Let us make remarks on two of these codes. Firstly, consider the $[[90,3,5]]$ code from \cref{table: search results3}. The three cyclic groups $C_5,C_3,C_2$ have coprime order, so in fact we have that  $C_5\times C_3\times C_2\cong C_{30}$. We can define an isomorphism on the generators
\begin{subequations}
    \begin{align}
        \phi: C_{5}\times C_3\times C_2&\to C_{30}\\
        \phi(x^iy^jz^k)&\mapsto a^{6i+10j+15k},
    \end{align}
\end{subequations}
where $x,y,z$ are the generators of the $C_5,C_3,$ and $C_2$ groups respectively; and $a$ is the generator of the $C_{30}$ cyclic group. Mapping the check polynomials, we find a $[[90,3,5]]$ code that has preserved the code parameters, with check polynomials $1+a^6,1+a^{16},1+a^{17}$. Monomials in the check polynomials do not have orders coprime to the group order, and so an equivalent code was not found when searching for codes presented in \cref{table: search results2}. Again, this suggests that there exist less restrictive search parameters for quantum codes than those we have considered in this work. For example, the check polynomial $1+a^6$ defines 6 disjoint copies of a classical distance 5 repetition code, i.e. a $[30,6,5]$ code. Results suggest that looking only for repetition codes in the univariate cyclic groups might not be the only way to obtain codes with good parameters. Indeed, it is a consequence of taking the balanced product $\otimes_R$, i.e. collapsing the dimensions or twisting the periodic boundary conditions, that lead to lower distances. The quantum codes might be out-performed by taking tensor products of other cyclic codes from the same group algebra. Detailed analysis of the code parameters would follow the work of \cite{liang2025generalizedtoriccodestwisted,eberhardt2024logicaloperatorsfoldtransversalgates}.

Secondly, let us consider the $[[6,3,2]]$ code from \cref{table: search results2}. After removing duplicated rows, this code has check matrices
\begin{subequations}
    \begin{align}
        H_X &= \begin{bmatrix}
            1 & 1 & 1 & 1 & 1 & 1
        \end{bmatrix},\\
        H_Z &= \begin{bmatrix}
            0 & 0 & 1 & 1 & 1 & 1\\
            1 & 1 & 0 & 0 & 1 & 1
        \end{bmatrix}.
    \end{align}
\end{subequations}
This is comparable to the $[[8,3,2]]$ color code, which admits a transversal $\operatorname{CCZ}$ gate. Both codes also involve a single X stabiliser on all qubits, and weight 4 Z stabilisers. However, the $[[6,3,2]]$ code has \textit{inter-code} non-Clifford gates on three copies of the code, whereas the colour code supports \textit{intra-code} non-Clifford gates on a single copy of the code.

\subsection{6-term group algebra codes}
For the rest of this section, we briefly touch on 6-term group algebra codes and the existence of pre-orientation for the 3-copy-cup gate. Conditions for pre-orientation are obtained by implementing the numerical algorithm described in \cref{pseudocode} -- that is, to frame the problem combinatorially as a perfect matching problem. First, note that the possible assignments for the weight 6 codes are given as 
\begin{align}
    (|\delta_{in}|,|\delta_{out}|,|\delta_{free}|)
    \in&\{(3,3,0),(2,4,0),(4,2,0),\notag\\
    &(1,5,0),(5,1,0),(2,2,2),\notag\\
    &(3,1,2),(1,3,2),(1,1,4)\notag\}.
\end{align}
Due to the scaling of the problem complexity, it proves difficult to even analyse the 6-term case as the number of terms and configurations grow exponentially. The algorithm lets us quickly rule out some of these assignments. Let us first analyse pre-orientation for the non-associative 3-copy-cup gate.
\begin{corollary}
    Consider a group algebra code defined by $\delta \in\mathbb{F}_2[G],|\delta|=6$, following the definitions of \cref{sec: group algebra}. The assignments
    \begin{equation}
        \{(1,5,0),(3,3,0),(3,1,2),(1,3,2)\}
    \end{equation}
    of group elements to coboundary partitions $\delta_{in}, \delta_{out},$ and $\delta_{free}$ for the non-associative 3-copy-cup gates do not have valid pre-orientation.
\end{corollary}
\begin{proof}
Validated numerically. However, some of these assignments are immediately shown not to have pre-orientation due to \cref{lemma: expansion}. For example, the assignments $(1,5,0),(1,3,2),(3,1,2)$ cannot satisfy pre-orientation due to singular conditions.
\end{proof}
We are left with the assignments $\{(2,4,0),(4,2,0),(2,2,2),(5,1,0)\}$. Let us focus on the $(2,2,2)$ assignment in the next theorem.
\begin{theorem} \label{thm: 222 group 6 assignment}
    Consider a 6-term group algebra code with $\delta=g_1+g_2+g_3+g_4+g_5+g_6\in \mathbb{F}_2[G]$, following the definitions of \cref{sec: group algebra}. The code has pre-orientation satisfying the non-associative 3-copy-cup gate conditions of \cref{prop: lambda 3 conditions} dependent on group elements assigned to partitions $\delta_{in},\delta_{out}$, and $\delta_{free}$. One such pre-orientation satisfies the conditions
    \begin{subequations} \label{eqn: 6 term CCZ}
        \begin{align}
            g_1^{-1}g_2&=g_2^{-1}g_1, \\
            g_3^{-1}g_4&=g_4^{-1}g_3, \\
            g_5^{-1}g_6 &= g_6^{-1}g_5, \\
            g_1g_2^{-1}= g_3g_4^{-1}&=g_5g_6^{-1},\\
        \delta_{in}&:=g_1+g_2,\notag\\
        \delta_{out}&:=g_3+g_4,\notag\\
        \delta_{free}&:=g_5+g_6.\notag
        \end{align}
    \end{subequations}
\end{theorem}
We find numerically that this is the only configuration satisfying the $(2,2,2)$ assignment. Moreover, the conditions are similar to \cref{eqn: 4 term CCZ} and imposes similar conditions on involution for the two new group elements, which have been assigned to $\delta_{free}$. It is clear that, as the number of terms grow, it is also more difficult to find valid configurations since there are more tuples and triples of intersections to pair, which can lead to contradictory conditions. For this particular case, we can also show that increasing the check weight does not lead to codes with better distance.

\begin{corollary}
    Abelian weight 6 group algebra codes with assignment described in \cref{thm: 222 group 6 assignment} have an upper-bounded code distance of 2. 
\end{corollary}
\begin{proof}
    The proof is the same as that for the weight 4 codes, given in \cref{thm: 220 max dist 2}. Given that
    \begin{subequations}
        \begin{align}
            g_4&=g_2g_1^{-1}g_3\\
            g_6 &= g_2g_1^{-1}g_5,
        \end{align}
    \end{subequations}
we can rearrange the group elements in the check element as
\begin{subequations}
    \begin{align}
        &g_1+g_2+g_3+g_4+g_5+g_6\\
        &=g_1(1+g_1^{-1}g_2+g_1^{-1}g_3+g_1^{-1}g_4+g_1^{-1}g_5+g_1^{-1}g_6)\\
        &= g_1(1+g_1^{-1}g_2 + g_1^{-1}g_3+g_1^{-1}g_2g_1^{-1}g_3+g_1^{-1}g_5+g_1^{-1}g_2g_1^{-1}g_5)\\
        &= g_1(1+g_1^{-1}g_2 +(1+g_1^{-1}g_2)g_1^{-1}g_3 + (1+g_1^{-1}g_2)(g_1^{-1}g_5))\\
        &= g_1(1+g_1^{-1}g_2)(1+g_1^{-1}g_3+g_1^{-1}g_5)
    \end{align}
\end{subequations}
If the group is abelian, then $(1+g_1^{-1}g_2)$ will be a valid codeword since $g_1^{-1}g_2$ is an involution. Therefore if the group is abelian, the code distance upper bound is 2.
\end{proof}
We numerically investigate the $(2,4,0)$ assignment as well, and find 315 valid configurations. However, upon closer inspection, we find that all configurations contain similar involutions as that of \cref{eqn: 6 term CCZ}; hence we omit the specific configurations here. We do not investigate the $(4,2,0)$ or $(5,1,0)$ assignments as the involve a prohibitively large number of terms after expanding the intersections. Moreover, the large number of terms will result in severely constrained group algebra codes.

A similar analysis can be conducted for the weight 6 group algebra codes satisfying the symmetric 3-copy-cup gate. The different conditions lead to different possible assignments and constraints on the group algebra code. However, just as in the non-associative cup product, there are many restrictions due to the larger number of terms that need to be paired. For example, unlike the non-associative cup product, there exist valid $(3,3,0)$ assignments. Two examples of configurations that satisfy the assignment return the conditions
\begin{subequations}
    \begin{align}
            g_1^{-1}g_2&=g_2^{-1}g_1\\
            g_1^{-1}g_2&=g_4^{-1}g_5\\
            g_1^{-1}g_3&=g_3^{-1}g_1\\
            g_1^{-1}g_3&=g_4^{-1}g_6\\
            g_2^{-1}g_3&=g_3^{-1}g_2\\
            g_2^{-1}g_3&=g_5^{-1}g_6\\
            g_4^{-1}g_5&=g_5^{-1}g_4\\
            g_4^{-1}g_6&=g_6^{-1}g_4\\
            g_5^{-1}g_6&=g_6^{-1}g_5
        \end{align}
    \end{subequations}
and 
\begin{subequations}
    \begin{align}
            g_1^{-1}g_2&=g_2^{-1}g_1\\
            g_1^{-1}g_2&=g_4^{-1}g_6\\
            g_1^{-1}g_3&=g_4^{-1}g_5\\
            g_2^{-1}g_3&=g_6^{-1}g_5\\
            g_4^{-1}g_6&=g_6^{-1}g_4,
        \end{align}
    \end{subequations}
where
\begin{subequations}
    \begin{align*}
        \delta_{in}&:=g_1+g_2+g_3,\\
        \delta_{out}&:=g_4+g_5+g_6.
    \end{align*}
\end{subequations}
These configurations are found numerically through a perfect matching algorithm. We do not implement the numerical search for weight 6 codes satisfying pre-orientation due to the severely restrictive conditions. However, we note that, just as in the case of weight 4 codes, results are suggestive that non-abelian group algebra codes could possibly have better parameters.

\section{Conclusion} \label{sec: conclusion}
In this work, we have determined the constraints on classical group algebra codes with increasing check weight so that they satisfy pre-orientation for the 2- and 3-copy-cup gates. When used to construct tensor product cochain complexes, these define quantum codes with constant-depth inter-code $\operatorname{CZ}$ and $\operatorname{CCZ}$ gates respectively. We have shown that determining these constraints can be formulated as a perfect matching problem, and that the number of constraints increase with the code check weight. 
Determining the set of all constraints is computationally challenging, and by implementing a perfect-matching algorithm, we both check the analysis of conditions for group algebra codes up to weight 5, and also determine examples of conditions for weight 6 group algebra codes to have pre-orientation for the 3-copy-cup gate. By utilising both numerics and analysis, we have shown that increased constraints result in more involution terms of pairs of group elements, which are particularly strict for abelian group algebra codes, and lead to low distances.

Through understanding subtleties of the integrated Leibniz rule, we have also extended the discussion of copy-cup gates to quantum codes constructed from classical codes with odd weight. In turn, this has allowed us to analyse a wide variety of quantum codes, including the bivariate bicycle codes of \cite{bravyi2023highthreshold}. We have shown that they do not satisfy pre-orientation for the 2- and 3-copy-cup gate, and that, in fact, no weight 3 or weight 5 group algebra codes can satisfy pre-orientation for the 3-copy-cup gate. This analysis has been supplemented by a numerical search for codes, which has led to new quantum codes such as a $[[144,4,12]]$ code with check weight 6, and a $[[72,14,6]]$ code with check weight 8, both equipped with constant-depth $\operatorname{CZ}$ gates.

Our analysis has also allowed us to reason the low distances of certain trivariate tricycle codes, as well as unify and extend examples of 2-fold and 3-fold balanced product group algebra codes with copy-cup gates described in previous works \cite{breuckmann2024cupsgatesicohomology,jacob2025singleshotdecodingfaulttolerantgates,menon2025magictricyclesefficientmagic}.
We have extended the discussion of copy-cup gates towards non-abelian group algebra codes, and have explicitly listed all possible conditions on weight 4 group algebra codes for the 2- and 3-copy cup gates. Parts of these conditions have been brought up in \cite{breuckmann2024cupsgatesicohomology,menon2025magictricyclesefficientmagic,jacob2025singleshotdecodingfaulttolerantgates} in various forms. We have also determined an upper distance bound on weight 4 abelian group algebra codes that satisfy pre-orientation for the non-associative 3-copy-cup gate, which has been empirically observed in \cite{menon2025magictricyclesefficientmagic,jacob2025singleshotdecodingfaulttolerantgates}, and have showed it is a result of involution conditions on the group elements.

Finally, we have also described balanced product cochain complexes constructed from abelian group algebra codes beyond trivariate abelian group algebras. We have explicitly utilised the ring isomorphism $R\otimes_RR\cong R$, previously described in e.g.~\cite[Equation 89]{menon2025magictricyclesefficientmagic} and \cite[Appendix A]{eberhardt2024logicaloperatorsfoldtransversalgates}, as well as the fundamental theorem of finite abelian groups \cite[Chapter 5.2 Thm.\ 3]{Dummit_Foote_2004} to construct complexes from general abelian group algebras, and on which we define quantum codes. Extending this construction to hypercube complexes as the balanced product of group algebra codes is straightforward.
We have illustrated a benefit of extending beyond trivariate abelian groups by searching for codes defined on balanced products of univariate cyclic group algebras, finding balanced product quantum codes constructed from weight 2 repetition codes that outperform similar trivariate codes from \cite{jacob2025singleshotdecodingfaulttolerantgates} in terms of distance.

Let us make two final remarks on future work. Firstly, it is clear from our results that increasing the check weight of the group algebra codes leads to more restrictive constraints for pre-orientation. On the other hand, increasing the check weight beyond 2 allows us to identify more classical codes beyond variants of the repetition code, leading potentially to an increased encoding rate. There is also numerical work, for example in \cite{tiew2024lowoverheadentanglinggatesgeneralised}, showing that quantum codes with better parameters can be found when check weight increases. There is a tradeoff between parameters of the quantum code and the strict conditions that come from group algebra codes with higher weight, which can also be affected by how non-associativity in the cup product is handled. Therefore an avenue for future work could be to identify how this tradeoff can be optimised.

Secondly, in this work, we have conducted numerical searches for abelian group algebra codes. However, the conditions we have derived can also be applied immediately to non-abelian group algebra codes. Constructing quantum codes then requires care in the left and right multiplication of group elements. Examples of quantum codes constructed from non-abelian groups can be found in \cite{pryadko_2BGA,Guo_2026}, but are currently limited to two-fold tensor products and square complexes. An immediate extension to our work would be to conduct a numerical search for non-abelian group algebra codes and the derived 2-fold tensor product quantum codes with non-trivial 2-copy-cup gates. A further extension would be to construct cube complexes from non-abelian groups, to obtain quantum codes with non-trivial 3-copy-cup gates; this can be generalised to higher-dimensional complexes.
\section*{Acknowledgments}
We would like to thank Abraham Jacob and Campbell McLauchlan for helpful discussions on the trivariate tricycle codes, in particular on the distance 2 property of the codes. 
RT acknowledges support from UK EPSRC (EP/SO23607/1).
Part of this work was conducted while the authors were visiting the Simons Institute for the Theory of Computing supported by DOE
QSA grant \#FP00010905.
\appendix 
\section{Details on Quantum Circuit}
The copy-cup gate of \cref{defn: copy-cup gate} is numerically simulated on the codes that are found. In this work we focus on the 2- and 3-copy-cup gates. Both simulations proceed in the same way. The integral of the cup product, defined in \cref{eqn: cohom inv} and restated here as
\begin{equation} \notag
    \int x_1\cup \cdots \cup x_\Lambda
\end{equation}
is evaluated for every tuple or triple of 1-cochains to obtain a list of $\operatorname{CZ}$ or $\operatorname{CCZ}$ gates for the copy-cup gate. The calculation is detailed for the 2-copy-cup gate in \cref{appendix: 2 copy-cup gate}, and for the 3-copy-cup gate in \cref{appendix: non associative 3 copy-cup gate}. Care needs to be taken when calculating the cup product for the balanced product cochain complex, with emphasis placed on calculating the cup products on coinvariants. A $\operatorname{CZ}$ or $\operatorname{CCZ}$ gate is applied on tuples and triples of qubits between codes if \cref{eqn: cohom inv} evaluated on the qubits is $1\mod 2$, depending on whether the 2- or 3-copy-cup gate is being applied.

After obtaining the set of gates in the copy-cup circuit, we check that the circuit results in non-trivial logical action on the cohomology representatives. This is done by counting the number of overlaps between the copy-cup circuit gates and the cohomology basis, and is detailed in \cref{alg: 2copy nontrivial}.

\begin{algorithm}[H]
\caption{2-copy-cup Logical action}
\begin{algorithmic}[1] \label{alg: 2copy nontrivial}
\Require Cohomology basis matrix $H^1_b$, with rows indicating the support of $\operatorname{X}$ logical operators; List $L$ of $\operatorname{CZ}$ gates in the copy-cup circuit of the form $(i,j)$ indicating $\operatorname{CZ}$ on qubit $i$ in code 1 and $j$ in code 2
\For{each $(x_1,x_2)$ in $\mathrm{rows}(H^1_b)\times\mathrm{rows}(H^1_b)$} \Comment{X logical operators in code 1 and code 2}
    \State Initialise overlap $\gets 0$  \Comment{Overlap between logical operators under the copy-cup circuit}
    \For{each $q_1$ in $x_1$ and $q_2$ in $x_2$} \Comment{For each tuple of qubits}
        \For{ $(i,j)$ in $L$} \Comment{Apply $\operatorname{CZ}$s in circuit}
            \If{$(q_1,q_2)==(i,j)$}
                \State overlap $\gets$ overlap$+1 \mod 2$ 
            \EndIf
        \EndFor
    \EndFor
        
    \If{overlap == $1\mod 2$}
        \State \Return \texttt{True}
    \EndIf
\EndFor
\State \Return \texttt{False}
\end{algorithmic}
\end{algorithm}

We also state the algorithm to check for constant-depth $\operatorname{CCZ}$ for the 3-copy-cup gate in \cref{alg: ccz}. This is the same as the discussion in \cite[Appendix C.4]{jacob2025singleshotdecodingfaulttolerantgates}. For every triple of logical X operators $(x_1,x_2,x_3)$, we count the number of $\operatorname{CCZ}$ gates in the circuit that intersect their support. If this intersection is odd, then there is a non-trivial logical $\operatorname{CCZ}$ gate. We check for presence of \emph{a} non-trivial $\operatorname{CCZ}$ by computing \emph{a} cohomology basis, not necessarily minimum weight.
\begin{algorithm}[H]
\caption{3-copy-cup Logical action}
\begin{algorithmic}[1] \label{alg: ccz}
\Require Cohomology basis matrix $H^1_b$, with rows indicating the support of $\operatorname{X}$ logical operators; List $L$ of $\operatorname{CZ}$ gates in the copy-cup circuit of the form $(i,j,k)$ indicating $\operatorname{CZ}$ on qubit $i,j,k$ in code 1,2, and 3 respectively

\For{each $(x_1,x_2,x_3)$ in $\mathrm{rows}(H^1_b)\times \mathrm{rows}(H^1_b)\times \mathrm{rows}(H^1_b)$} \Comment{X logical operator in code 1, code 2, and code 3}
    \State Initialise overlap $\gets 0$  \Comment{Overlap between logical operators under the copy-cup circuit}
    \For{each $(q_1,q_2,q_3)\in x_1\times x_2 \times x_3$} \Comment{For every triple of qubits}
        \For{ $(i,j,k)$ in $L$} \Comment{Apply $\operatorname{CCZ}$s in circuit}
            \If{$(q_1,q_2,q_3)==(i,j,k)$}
                \State overlap $\gets$ overlap$+1 \mod 2$
            \EndIf
        \EndFor
    \EndFor
    \If{overlap == $1\mod 2$}
        \State \Return \texttt{True}
    \EndIf
\EndFor
\State \Return \texttt{False}
\end{algorithmic} 
\end{algorithm}
\section{Cup product of balanced product cochains}
We explicitly go through the cup product calculation for the tensor product over the ring $R$, i.e.\ the balanced product $\otimes_R$. The hypergraph product is the trivial case where $R=\mathbb{F}_2$.

The cup product calculation for the balanced product requires one to compute with coinvariants. The partitioning of cochains into equivalence classes introduces the relation
\begin{equation}
    ar\otimes_Rb \equiv a\otimes_R rb, r\in R
\end{equation}
for cochains $a,b$ where ring elements can be brought across the tensor product. To determine the cup product between two balanced product cochains, one must compute all non-zero cup products between a single representative cochain from one equivalence class, and all representatives in the other. When the integral is taken to be the Hamming weight, it follows that there is a non-zero cup product only if a cochain has an odd number of non-zero cup products with the coinvariants of another cochain. Further information on computing with coinvariants can be found in \cite[Section 3.3]{breuckmann2024cupsgatesicohomology}

We will illustrate this principle with cup products for the 2- and 3-copy-cup gates in the following sections, and see that the computation essentially involves summing over the coinvariants.

\subsection{Cup product calculation for $\Lambda=2$} \label{appendix: 2 copy-cup gate}
The case $\Lambda=2$ is lifted and further extended from \cite[Section 5.3.1]{breuckmann2024cupsgatesicohomology}. Non-zero cup products are only between cochains of the form $p^1\otimes 1^0$ and $1^0\otimes q^1$ with superscripts indicating whether it is a 0- or 1-cochain. Let us consider taking the tensor product over $R$ where $R$ is the entire ring the group algebra code is defined on. Then we can choose a basis of 1-cochains as $p\otimes_R1\in C^1\otimes C^0$ and $1\otimes_Rq\in C^0\otimes C^1$. Then the calculation goes as
\begin{subequations}
    \begin{align*}
        (p\otimes_R1)\cup (1\otimes_R q) &= \sum_g(p\otimes_{\mathbb{F}_2}1)\cup (g^{-1}\otimes_{\mathbb{F}_2} gq)\\
        &=\sum_g(p\cup g)\otimes(1\cup gq)\\
        &= \sum_g \delta_{p\in \delta^1_{in}(g)}\delta_{gq\in\delta^2_{out}} p\otimes gq\\
        &= \sum_g \delta_{g\in \delta^{-1}_{in}(p)} \delta_{g\in\delta^2_{out}(q^{-1})} p\otimes gq\\
        &= \sum_g \delta_{g\in \delta^{-1}_{in}(p)\cap \delta^2_{out}(q^{-1})}p\otimes gq.\\
    \end{align*}
\end{subequations}
We explicitly write out the tensor products in the first line of the equation, and omit it for the rest of the calculation. Note the distinction between delta functions and the coboundary operators. The notation $\delta^{-i}$ refers to the coboundary operator $\delta$ of the $i^{th}$ classical code in the tensor product cochain complex, with the negative superscript then being the antipode map $a\mapsto a^{-1}$ being applied to all group elements in this partition. If the number of terms in this sum over $g$ is $1\mod2$, then there is a $\operatorname{CZ}$ gate from the cochain $p\otimes_R1$ of code 1 to the cochain $1\otimes_R q$ of code 2. Similarly, we omit the subscripts on tensor products and compute 
\begin{subequations}
    \begin{align*}
        (1\otimes q)\cup (p\otimes1) &= \sum_g(1\otimes q)\cup (pg^{-1}\otimes g)\\
        &=\sum_g(1\cup pg^{-1}) \otimes (q\cup g)\\
        &= \sum_g \delta_{pg^{-1}\in \delta^1_{out}} \delta_{q\in \delta^2_{in}(g)} pg^{-1}\otimes q\\
        &= \sum_g \delta_{g^{}\in \delta^{-1}_{out}(p)} \delta_{g\in \delta^{-2}_{in}(q)} pg^{-1}\otimes q\\
        &= \sum_g \delta_{g\in \delta^{-1}_{out}(p)\cap \delta^{-2}_{in}(q^{})}pg^{-1}\otimes q.
    \end{align*}
\end{subequations}
If the number of terms in this sum over $g$ is $1\mod2$, then there is a gate from the cochain $1\otimes_R q$ of code 1 to the cochain $p\otimes_R1$ of code 2. Now let us go to the 3-copy-cup gate.
\subsection{Cup product calculation for $\Lambda=3$} \label{appendix: non associative 3 copy-cup gate}
The non-associative and symmetric 3-copy-cup gates have different calculations due to the order the cup product is applied. We show the explicit calculation for the non-associative copy-cup gate. There are 6 calculations to go through, but each ends up as a sum over group elements with delta functions. For the cup product between $p\otimes 1\otimes 1,1\otimes q\otimes 1$ and $1\otimes 1 \otimes r$, we have
\begin{subequations} \label{asd}
    \begin{align}
        &(p\otimes 1\otimes 1)\cup (1\otimes q\otimes 1) \cup (1\otimes 1 \otimes r)\notag \\
        &= \sum_{g_1,g_2,g_3,g_4} (p\otimes 1\otimes 1)\cup (g_1^{-1}\otimes g_1qg_2^{-1}\otimes g_2) \cup (g_3^{-1}\otimes g_3g_4^{-1} \otimes g_4r)\label{asd1}\\
        &=\sum_{g_1,g_2,g_3,g_4} (p\cup g_1^{-1}\cup g_3^{-1}) \otimes (1\cup g_1qg_2^{-1}\cup g_3g_4^{-1})\otimes (1\cup g_2\cup g_4r)\label{asd2}\\
        &=\sum_{g_1,g_2,g_3,g_4} \left(\delta_{p\in\delta^1_{in}(g_1^{-1})}\delta_{p\in\delta^1_{in}(g_3^{-1})}\delta_{g_1qg_2^{-1}\in\delta^2_{out}}
        \delta_{g_1qg_2^{-1}\in\delta^2_{in}g_3g^{-1}_4}\delta_{g_4r\in\delta^3_{out}g_2}\delta_{1,g_2}\right) p\otimes g_1qg_2^{-1}\otimes g_4r\label{asd3}\\
        &= \sum_{g_1,g_3,g_4} \left(\delta_{g_1\in p^{-1}\delta^1_{in}}\delta_{g_3\in p^{-1}\delta^1_{in}}\delta_{g_1q\in\delta^2_{out}}
        \delta_{g_1q\in\delta^2_{in}g_3g^{-1}_4}\delta_{g_4r\in\delta^3_{out}}\right) p\otimes g_1q\otimes g_4r\label{asd4}\\
        &=\sum_{g_1,g_3,g_4} \left(\delta_{g_1\in p^{-1}\delta^1_{in}}\delta_{g_3\in p^{-1}\delta^1_{in}}\delta_{g_1\in\delta^2_{out}q^{-1}}
        \delta_{g_3\in \delta_{in}^{-2}g_1qg_4}\delta_{g_4\in\delta^3_{out}r^{-1}}\right) p\otimes g_1q\otimes g_4r\label{asd5}\\
        &=\sum_{g_1,g_3,g_4} \left(\delta_{g_1\in p^{-1}\delta^1_{in}\cap \delta_{out}^2q^{-1}}\delta_{g_3\in p^{-1}\delta^1_{in}\cap \delta_{in}^{-2}g_1qg_4}
        \delta_{g_4\in\delta^3_{out}r^{-1}}\right) p\otimes g_1q\otimes g_4r\label{asd6}\\
        &=\sum_{g_1,g_4} \left(\delta_{g_1\in p^{-1}\delta^1_{in}\cap \delta_{out}^2q^{-1}}\delta_{g_4\in\delta^3_{out}r^{-1}}\right) \mid p^{-1}\delta^1_{in}\cap \delta_{in}^{-2}g_1qg_4\mid p\otimes g_1q\otimes g_4r.\label{asd7}
    \end{align}
\end{subequations}
In \cref{asd1}, the cochains are first converted into their coinvariant form and expanded, resulting in the sums. In \cref{asd2}, the cup product is distributed over the tensor product. Then, the cup products are evaluated to return delta functions. In \cref{asd4}, the sum over $g_2$ is killed by evaluating $\delta_{1,g_2}$. Other delta functions have their arguments re-arranged; for example, the condition $p\in\delta^1_{in}(g_1^{-1})$ in \cref{asd3} is rearranged to $g_1\in p^{-1}\delta^1_{in}$ in \cref{asd4} since $g_1$ is involved in a sum. In \cref{asd7}, the sum over $g_3$ is killed by the condition $g_3\in p^{-1}\delta^1_{in}\cap \delta_{in}^{-2}g_1qg_4$, returning the argument $\mid p^{-1}\delta^1_{in}\cap \delta_{in}^{-2}g_1qg_4\mid$ since the intersection term does not have any dependence on $g_3$. We attempt to simplify the calculation as much as possible by killing sums. Without simplification, finding interacting cochains in the 3-copy-cup product would involve looping over the group elements 7 times, which quickly becomes unfeasible even for small group orders. Furthermore, we note that by keeping track of left and right multiplication, this calculation applies to both abelian and non-abelian group algebra codes.

We present the final simplified equations for the 5 other non-zero cup products between cochains below. Simplification follows as in \cref{asd}, starting with 6 delta functions over 4 sums and attempting to kill off any sums. We omit the bulk of the calculation and present the final simplified sums used in numerical simulations for the non-associative 3-copy-cup gate.
\begin{subequations}
    \begin{align*}
        &(p\otimes 1\otimes 1)\cup (1\otimes 1 \otimes r) \cup (1\otimes q\otimes 1) \\
        &=\sum_{g_2,g_3,g_4}p\otimes g_3qg_4^{-1}\otimes g_2r \cdot
        \delta_{g_2\in p^{-1}\delta^{1}_{in}\cap \delta_{out}^3r^{-1}}
        \delta_{g_3\in p^{-1}\delta^{1}_{in}\cap \delta_{out}^2g_4q^{-1}}
        \delta_{g_4\in\delta_{in}^{-3}g_2r}.\\
        &(1\otimes q\otimes 1)\cup (1\otimes 1 \otimes r) \cup (p\otimes 1\otimes 1) \\
        &=\sum_{g_2,g_3}pg_3^{-1}\otimes q \otimes g_2r \cdot 
        |q^{-1}\delta_{in}^2 g_3\cap \delta_{in}^3g_2r| 
        \delta_{g_3\in \delta_{out}^{-1}p}
        \delta_{g_2\in q^{-1}\delta_{in}^2\cap\delta_{out}^3r^{-1}}.\\
        &(1\otimes q\otimes 1)\cup (p\otimes 1\otimes 1)\cup (1\otimes 1 \otimes r)  \\
        &=\sum_{g_1,g_4}pg_1^{-1}\otimes q \otimes g_4r \cdot 
        \delta_{g_1\in \delta_{out}^{-1}p\cap \delta_{in}^{-2}q}
        \delta_{g_4\in \delta_{out}^3r^{-1}}
        |g_1p^{-1}\delta_{in}^1\cap\delta_{in}^{-2}g_4q|.\\
        &(1\otimes 1 \otimes r) \cup (p\otimes 1\otimes 1)\cup (1\otimes q\otimes 1)  \\
        &=\sum_{g_1,g_3,g_4}pg_1^{-1}\otimes g_3qg_4^{-1} \otimes r \cdot 
        \delta_{g_1^{-1}\in p^{-1}\delta_{out}^1\cap p^{-1}\delta_{in}^1g_3^{-1}\cap r^{-1}\delta_{in}^3}
        \delta_{g_3qg_4^{-1}\in\delta_{out}^2}
        \delta_{g_4^{-1}\in r^{-1}\delta_{in}^3}\\
        &(1\otimes 1 \otimes r) \cup (1\otimes q\otimes 1)\cup (p\otimes 1\otimes 1) \\
        &=\sum_{g_2,g_3}pg_3^{-1}\otimes qg_2^{-1} \otimes r \cdot 
        \mid \delta_{in}^{-3}r\cap g_2q^{-1}\delta_{in}^2g_3\mid
        \delta_{g_3\in\delta_{out}^{-1}p}
        \delta_{g_2\in\delta_{in}^{-3}r\cap\delta_{out}^{-2}q}.
    \end{align*}
\end{subequations}
These conditions are used to find cochains involved in the copy-cup circuit. There is an analogous calculation for the symmetric triple cup product. In this case, we find that fixing the $1\otimes q\otimes 1$ cochain and expanding the coinvariants $\sum_{g_1,g_2}g_1^{-1}\otimes g_1g_2^{-1}\otimes g_2r$ and $\sum_{g_3,g_4}pg_3^{-1}\otimes g_3g_4^{-1}\otimes g_4$ lead to a particularly ordered simplified expression. Again, we can show an example calculation as
\begin{subequations} \label{asdf}
    \begin{align}
        &(p\otimes 1\otimes 1)\cup (1\otimes q\otimes 1) \cup (1\otimes 1 \otimes r)\notag \\
        &= \sum_{g_1,g_2,g_3,g_4} (pg_3^{-1}\otimes g_3g_4^{-1}\otimes g_4)\cup (1\otimes q\otimes 1) \cup (g_1^{-1}\otimes g_1g_2^{-1} \otimes g_2r)\label{asdf0}\\
        &= \sum_{g_1,g_2,g_3,g_4}
        (pg_3^{-1}\cup 1 \cup g_1^{-1})\otimes (g_3g_4^{-1}\cup q \cup g_1g_2^{-1})\otimes (g_4\cup 1\cup g_2r\label{asdf1})\\
        &= \sum_{g_1,g_2,g_3,g_4}
        pg_3^{-1}\otimes q \otimes g_2r \cdot
        \delta_{pg_3^{-1}\in \delta_{in}^1} 
        \delta_{pg_3^{-1}\in \delta_{in}^1g_1^{-1}}
        \delta_{q\in \delta_{out}^2g_3g_4^{-1}}
        \delta_{q\in\delta_{in}^2g_1g_2^{-1}}
        \delta_{g_2r\in \delta_{out}^3}
        \delta_{g_2r\in \delta_{out}^3g_4} \label{asdf2}\\
        &= \sum_{g_1,g_2,g_3,g_4}
        pg_3^{-1}\otimes q \otimes g_2r \cdot
        \delta_{g_3\in \delta_{in}^{-1}p}
        \delta_{g_1\in g_3p^{-1}\delta_{in}^1}
        \delta_{g_4\in q^{-1}\delta_{out}^2g_3}
        \delta_{g_1\in{\delta_{in}^{-2}qg_2}}
        \delta_{g_2\in\delta_{out}^3r^{-1}}
        \delta_{g_4\in\delta_{out}^{-3}g_2r}
        \label{asdf3}\\
        &= \sum_{g_2,g_3}
        pg_3^{-1}\otimes q \otimes g_2r \cdot
        | g_3p^{-1}\delta_{in}^1\cap \delta_{in}^{-2}qg_2| \cdot
        | q^{-1}\delta_{out}^2g_3\cap \delta_{out}^{-3}g_2r|\cdot
        \delta_{g_2\in\delta_{out}^3r^{-1}}        
        \delta_{g_3\in \delta_{in}^{-1}p}
        \label{asdf4}
    \end{align}
\end{subequations}
In \cref{asdf0} the cochains are expanded as coinvariants; in \cref{asdf1} the cup product is distributed over the tensor product; in \cref{asdf2} the cup products are evaluated with non-associativity handled by the symmetric rule, returning 6 delta functions; in \cref{asdf3} the delta functions are re-arranged so that the group elements $g_i$ are on the left-hand side, in preparation to kill sums; in \cref{asdf4}, the sums over $g_1,g_4$ are killed and replaced with cardinality of some intersections. All 6 calculations of non-zero cup products have a similar form, and we summarise the simplified equations below.
\begin{subequations}
    \begin{align*}
        &(p\otimes 1\otimes 1)\cup (1\otimes 1 \otimes r) \cup (1\otimes q\otimes 1) \\
        &=\sum_{g_2,g_3}pg_3^{-1}\otimes q\otimes g_2r \cdot
        |g_3p^{-1}\delta_{in}^1\cap\delta_{out}^{-2}qg_2|\cdot
        |\delta_{out}^{-3}g_2r\cap q^{-1}\delta_{out}^2g_3|\cdot
        \delta_{g_3\in\delta_{in}^{-1}p}
        \delta_{g_2\in\delta_{in}^{3}r^{-1}}\\
        &(1\otimes q\otimes 1)\cup (1\otimes 1 \otimes r) \cup (p\otimes 1\otimes 1) \\
        &=\sum_{g_2,g_3} pg_3^{-1}\otimes q \otimes g_2r \cdot
        |g_3p^{-1}\delta_{out}^1\cap \delta_{in}^{-2}qg_2|\cdot|\delta_{in}^{-3}g_2r\cap q^{-1}\delta_{in}^2g_3|\cdot
        \delta_{g_3\in \delta_{out}^{-1}p}
        \delta_{g_2\in \delta_{out}^3r^{-1}}\\  
        &(1\otimes q\otimes 1)\cup (p\otimes 1\otimes 1)\cup (1\otimes 1 \otimes r)  \\
        &=\sum_{g_2,g_3} pg_3^{-1}\otimes q \otimes g_2r \cdot
        |\delta_{out}^{-3}g_2r\cap q^{-1}\delta_{in}^2g_3|\cdot
        |\delta_{in}^{-2}qg_2\cap g_3p^{-1}\delta_{in}^1|\cdot
        \delta_{g_3\in \delta_{out}^{-1}p}
        \delta_{g_2\in\delta_{out}^3r^{-1}}\\
        &(1\otimes 1 \otimes r) \cup (p\otimes 1\otimes 1)\cup (1\otimes q\otimes 1)  \\
        &=\sum_{g_2,g_3} pg_3^{-1}\otimes q \otimes g_2r \cdot
        |q^{-1}\delta_{out}^2g_3\cap \delta_{in}^{-3}g_2r|\cdot
        |g_3p^{-1}\delta_{out}^1\cap \delta_{out}^{-2}qg_2|\cdot
        \delta_{g_3\in \delta_{in}^{-1}p}
        \delta_{g_2\in\delta_{in}^3r^{-1}}\\
        &(1\otimes 1 \otimes r) \cup (1\otimes q\otimes 1)\cup (p\otimes 1\otimes 1) \\
        &=\sum_{g_2,g_3} pg_3^{-1}\otimes q \otimes g_2r \cdot
        |q^{-1}\delta_{in}^2g_3\cap \delta_{in}^{-3}g_2r|\cdot
        |g_3p^{-1}\delta_{out}^1\cap \delta_{out}^{-2}qg_2|\cdot
        \delta_{g_3\in \delta_{out}^{-1}p}
        \delta_{g_2\in\delta_{in}^3r^{-1}}\\
    \end{align*}
\end{subequations}
\bibliography{apssamp}%

@PREAMBLE{
 "\providecommand{\noopsort}[1]{}" 
 # "\providecommand{\singleletter}[1]{#1}%" 
}

@misc{golowich2024quantumldpccodestransversal,
      title={Quantum LDPC Codes with Transversal Non-Clifford Gates via Products of Algebraic Codes}, 
      author={Louis Golowich and Ting-Chun Lin},
      year={2024},
      eprint={2410.14662},
      archivePrefix={arXiv},
      primaryClass={quant-ph},
      url={https://arxiv.org/abs/2410.14662}, 
}

@misc{pryadko_2BGA
,
      title={Quantum two-block group algebra codes}, 
      author={Hsiang-Ku Lin and Leonid P. Pryadko},
      year={2023},
      eprint={2306.16400},
      archivePrefix={arXiv},
      primaryClass={quant-ph}
}

@article{Edmonds_1965, title={Paths, Trees, and Flowers}, volume={17}, DOI={10.4153/CJM-1965-045-4}, journal={Canadian Journal of Mathematics}, author={Edmonds, Jack}, year={1965}, pages={449–467}}

@misc{higgott2021pymatchingpythonpackagedecoding,
      title={PyMatching: A Python package for decoding quantum codes with minimum-weight perfect matching}, 
      author={Oscar Higgott},
      year={2021},
      eprint={2105.13082},
      archivePrefix={arXiv},
      primaryClass={quant-ph},
      url={https://arxiv.org/abs/2105.13082}, 
}

@book{Dummit_Foote_2004, place={Hoboken}, edition={3}, title={Abstract algebra}, publisher={Wiley}, author={Dummit, David S. and Foote, Richard M.}, year={2004}}

@misc{lin2025abelianmulticyclecodessingleshot,
      title={Abelian multi-cycle codes for single-shot error correction}, 
      author={Hsiang-Ku Lin and Pak Kau Lim and Alexey A. Kovalev and Leonid P. Pryadko},
      year={2025},
      eprint={2506.16910},
      archivePrefix={arXiv},
      primaryClass={quant-ph},
      url={https://arxiv.org/abs/2506.16910}, 
}

@misc{perfect_matching, title={Perfect matching}, url={https://mathworld.wolfram.com/PerfectMatching.html}, journal={Mathworld}, 
year={},
author={Weisstein, Eric}}

@misc{jacob2025singleshotdecodingfaulttolerantgates,
      title={Single-Shot Decoding and Fault-tolerant Gates with Trivariate Tricycle Codes}, 
      author={Abraham Jacob and Campbell McLauchlan and Dan E. Browne},
      year={2025},
      eprint={2508.08191},
      archivePrefix={arXiv},
      primaryClass={quant-ph},
      url={https://arxiv.org/abs/2508.08191}, 
}

@misc{menon2025magictricyclesefficientmagic,
      title={Magic tricycles: efficient magic state generation with finite block-length quantum LDPC codes}, 
      author={Varun Menon and J. Pablo Bonilla-Ataides and Rohan Mehta and Daniel Bochen Tan and Mikhail D. Lukin},
      year={2025},
      eprint={2508.10714},
      archivePrefix={arXiv},
      primaryClass={quant-ph},
      url={https://arxiv.org/abs/2508.10714}, 
}

@misc{liang2025generalizedtoriccodestwisted,
      title={Generalized toric codes on twisted tori for quantum error correction}, 
      author={Zijian Liang and Ke Liu and Hao Song and Yu-An Chen},
      year={2025},
      eprint={2503.03827},
      archivePrefix={arXiv},
      primaryClass={quant-ph},
      url={https://arxiv.org/abs/2503.03827}, 
}

@article{Horsman_2012,
   title={Surface code quantum computing by lattice surgery},
   volume={14},
   ISSN={1367-2630},
   url={http://dx.doi.org/10.1088/1367-2630/14/12/123011},
   DOI={10.1088/1367-2630/14/12/123011},
   number={12},
   journal={New Journal of Physics},
   publisher={IOP Publishing},
   author={Horsman, Dominic and Fowler, Austin G and Devitt, Simon and Meter, Rodney Van},
   year={2012},
   month=dec, pages={123011} }

@misc{tiew2024lowoverheadentanglinggatesgeneralised,
      title={Low-Overhead Entangling Gates from Generalised Dehn Twists}, 
      author={Ryan Tiew and Nikolas P. Breuckmann},
      year={2024},
      eprint={2411.03302},
      archivePrefix={arXiv},
      primaryClass={quant-ph},
      url={https://arxiv.org/abs/2411.03302}, 
}

@article{Tillich_2014,
   title={Quantum LDPC Codes With Positive Rate and Minimum Distance Proportional to the Square Root of the Blocklength},
   volume={60},
   ISSN={1557-9654},
   url={http://dx.doi.org/10.1109/TIT.2013.2292061},
   DOI={10.1109/tit.2013.2292061},
   number={2},
   journal={IEEE Transactions on Information Theory},
   publisher={Institute of Electrical and Electronics Engineers (IEEE)},
   author={Tillich, Jean-Pierre and Zemor, Gilles},
   year={2014},
   month=feb, pages={1193–1202} }

@misc{bravyi2023highthreshold,
      title={High-threshold and low-overhead fault-tolerant quantum memory}, 
      author={Sergey Bravyi and Andrew W. Cross and Jay M. Gambetta and Dmitri Maslov and Patrick Rall and Theodore J. Yoder},
      year={2023},
      eprint={2308.07915},
      archivePrefix={arXiv},
      primaryClass={quant-ph}
}

@misc{cowtan2026parallellogicalmeasurementsquantum,
      title={Parallel Logical Measurements via Quantum Code Surgery}, 
      author={Alexander Cowtan and Zhiyang He and Dominic J. Williamson and Theodore J. Yoder},
      year={2026},
      eprint={2503.05003},
      archivePrefix={arXiv},
      primaryClass={quant-ph},
      url={https://arxiv.org/abs/2503.05003}, 
}

@article{Cohen_2022,
   title={Low-overhead fault-tolerant quantum computing using long-range connectivity},
   volume={8},
   ISSN={2375-2548},
   url={http://dx.doi.org/10.1126/sciadv.abn1717},
   DOI={10.1126/sciadv.abn1717},
   number={20},
   journal={Science Advances},
   publisher={American Association for the Advancement of Science (AAAS)},
   author={Cohen, Lawrence Z. and Kim, Isaac H. and Bartlett, Stephen D. and Brown, Benjamin J.},
   year={2022},
   month=may }

@misc{breuckmann2024cupsgatesicohomology,
      title={Cups and Gates I: Cohomology invariants and logical quantum operations}, 
      author={Nikolas P. Breuckmann and Margarita Davydova and Jens N. Eberhardt and Nathanan Tantivasadakarn},
      year={2024},
      eprint={2410.16250},
      archivePrefix={arXiv},
      primaryClass={quant-ph},
      url={https://arxiv.org/abs/2410.16250}, 
}

@article{Berman_1969, title={On the theory of group codes}, volume={3}, DOI={10.1007/bf01072842}, number={1}, journal={Cybernetics}, author={Berman, S. D.}, year={1969}, pages={25–31}}

@ARTICLE{6770805,
  author={Williams, F. J. Mac},
  journal={The Bell System Technical Journal}, 
  title={Binary codes which are ideals in the group algebra of an abelian group}, 
  year={1970},
  volume={49},
  number={6},
  pages={987-1011},
  doi={10.1002/j.1538-7305.1970.tb01812.x}}

@misc{jitman2010checkablecodesgrouprings,
      title={Checkable Codes from Group Rings}, 
      author={Somphong Jitman and San Ling and Hongwei Liu and Xiaoli Xie},
      year={2010},
      eprint={1012.5498},
      archivePrefix={arXiv},
      primaryClass={cs.IT},
      url={https://arxiv.org/abs/1012.5498}, 
}

@article{PhysRevX.15.021065,
  title = {Fast and Parallelizable Logical Computation with Homological Product Codes},
  author = {Xu, Qian and Zhou, Hengyun and Zheng, Guo and Bluvstein, Dolev and Ataides, J. Pablo Bonilla and Lukin, Mikhail D. and Jiang, Liang},
  journal = {Phys. Rev. X},
  volume = {15},
  issue = {2},
  pages = {021065},
  numpages = {34},
  year = {2025},
  month = {May},
  publisher = {American Physical Society},
  doi = {10.1103/PhysRevX.15.021065},
  url = {https://link.aps.org/doi/10.1103/PhysRevX.15.021065}
}

@misc{ide2024faulttolerantlogicalmeasurementshomological,
      title={Fault-tolerant logical measurements via homological measurement}, 
      author={Benjamin Ide and Manoj G. Gowda and Priya J. Nadkarni and Guillaume Dauphinais},
      year={2024},
      eprint={2410.02753},
      archivePrefix={arXiv},
      primaryClass={quant-ph},
      url={https://arxiv.org/abs/2410.02753}, 
}

@misc{huang2022homomorphiclogicalmeasurements,
      title={Homomorphic Logical Measurements}, 
      author={Shilin Huang and Tomas Jochym-O'Connor and Theodore J. Yoder},
      year={2022},
      eprint={2211.03625},
      archivePrefix={arXiv},
      primaryClass={quant-ph},
      url={https://arxiv.org/abs/2211.03625}, 
}

@misc{he2025extractorsqldpcarchitecturesefficient,
      title={Extractors: QLDPC Architectures for Efficient Pauli-Based Computation}, 
      author={Zhiyang He and Alexander Cowtan and Dominic J. Williamson and Theodore J. Yoder},
      year={2025},
      eprint={2503.10390},
      archivePrefix={arXiv},
      primaryClass={quant-ph},
      url={https://arxiv.org/abs/2503.10390}, 
}

@misc{cross2024improvedqldpcsurgerylogical,
      title={Improved QLDPC Surgery: Logical Measurements and Bridging Codes}, 
      author={Andrew Cross and Zhiyang He and Patrick Rall and Theodore Yoder},
      year={2024},
      eprint={2407.18393},
      archivePrefix={arXiv},
      primaryClass={quant-ph},
      url={https://arxiv.org/abs/2407.18393}, 
}

@article{Breuckmann_2021_2,
	doi = {10.1109/tit.2021.3097347},
  
	url = {https://doi.org/10.1109%2Ftit.2021.3097347},
  
	year = 2021,
	month = {oct},
  
	publisher = {Institute of Electrical and Electronics Engineers ({IEEE})},
  
	volume = {67},
  
	number = {10},
  
	pages = {6653--6674},
  
	author = {Nikolas P. Breuckmann and Jens N. Eberhardt},
  
	title = {Balanced Product Quantum Codes},
  
	journal = {{IEEE} Transactions on Information Theory}
}

@article{Guo_2026,
   title={Toward Self-Correcting Quantum Codes for Neutral Atom Arrays},
   volume={7},
   ISSN={2691-3399},
   url={http://dx.doi.org/10.1103/mfmt-fwkg},
   DOI={10.1103/mfmt-fwkg},
   number={1},
   journal={PRX Quantum},
   publisher={American Physical Society (APS)},
   author={Guo, Jinkang and Hong, Yifan and Kaufman, Adam and Lucas, Andrew},
   year={2026},
   month=jan }

@misc{symons2025sequencesbivariatebicyclecodes,
      title={Sequences of Bivariate Bicycle Codes from Covering Graphs}, 
      author={Benjamin C. B. Symons and Abhishek Rajput and Dan E. Browne},
      year={2025},
      eprint={2511.13560},
      archivePrefix={arXiv},
      primaryClass={quant-ph},
      url={https://arxiv.org/abs/2511.13560}, 
}

@misc{mian2026multivariatemulticyclecodescomplete,
      title={Multivariate Multicycle Codes for Complete Single-Shot Decoding}, 
      author={Feroz Ahmed Mian and Owen Gwilliam and Stefan Krastanov},
      year={2026},
      eprint={2601.18879},
      archivePrefix={arXiv},
      primaryClass={quant-ph},
      url={https://arxiv.org/abs/2601.18879}, 
}

@article{Campbell_2019,
   title={A theory of single-shot error correction for adversarial noise},
   volume={4},
   ISSN={2058-9565},
   url={http://dx.doi.org/10.1088/2058-9565/aafc8f},
   DOI={10.1088/2058-9565/aafc8f},
   number={2},
   journal={Quantum Science and Technology},
   publisher={IOP Publishing},
   author={Campbell, Earl T},
   year={2019},
   month=feb, pages={025006} }

@misc{nguyen2024goodbinaryquantumcodes,
      title={Good binary quantum codes with transversal CCZ gate}, 
      author={Quynh T. Nguyen},
      year={2024},
      eprint={2408.10140},
      archivePrefix={arXiv},
      primaryClass={quant-ph},
      url={https://arxiv.org/abs/2408.10140}, 
}

@misc{he2025asymptoticallygoodquantumcodes,
      title={Asymptotically Good Quantum Codes with Addressable and Transversal Non-Clifford Gates}, 
      author={Zhiyang He and Vinod Vaikuntanathan and Adam Wills and Rachel Yun Zhang},
      year={2025},
      eprint={2507.05392},
      archivePrefix={arXiv},
      primaryClass={quant-ph},
      url={https://arxiv.org/abs/2507.05392}, 
}

@article{Breuckmann_2021,
	doi = {10.1103/prxquantum.2.040101},
  
	url = {https://doi.org/10.1103%2Fprxquantum.2.040101},
  
	year = 2021,
	month = {oct},
  
	publisher = {American Physical Society ({APS})},
  
	volume = {2},
  
	number = {4},
  
	author = {Nikolas P. Breuckmann and Jens Niklas Eberhardt},
  
	title = {Quantum Low-Density Parity-Check Codes},
  
	journal = {{PRX} Quantum}
}

@misc{eberhardt2024logicaloperatorsfoldtransversalgates,
      title={Logical Operators and Fold-Transversal Gates of Bivariate Bicycle Codes}, 
      author={Jens Niklas Eberhardt and Vincent Steffan},
      year={2024},
      eprint={2407.03973},
      archivePrefix={arXiv},
      primaryClass={quant-ph},
      url={https://arxiv.org/abs/2407.03973}, 
}

@incollection{GOTTESMAN2006196,
title = {Quantum Error Correction and Fault Tolerance},
editor = {Jean-Pierre Françoise and Gregory L. Naber and Tsou Sheung Tsun},
booktitle = {Encyclopedia of Mathematical Physics},
publisher = {Academic Press},
address = {Oxford},
pages = {196-201},
year = {2006},
isbn = {978-0-12-512666-3},
doi = {https://doi.org/10.1016/B0-12-512666-2/00273-X},
url = {https://www.sciencedirect.com/science/article/pii/B012512666200273X},
author = {D. Gottesman}
}

@misc{swaroop2024universaladaptersquantumldpc,
      title={Universal adapters between quantum LDPC codes}, 
      author={Esha Swaroop and Tomas Jochym-O'Connor and Theodore J. Yoder},
      year={2024},
      eprint={2410.03628},
      archivePrefix={arXiv},
      primaryClass={quant-ph},
      url={https://arxiv.org/abs/2410.03628}, 
}
\end{document}